\documentclass[a4paper,USenglish,cleveref,autoref,thm-restate]{lipics-v2021}

\pdfoutput=1

\usepackage[dvipsnames]{xcolor}

\usepackage{etoolbox}

\usepackage{graphicx}
\usepackage{amsmath}
\usepackage{amssymb}
\usepackage{relsize}
\usepackage{bussproofs} \EnableBpAbbreviations

\usepackage[noxy]{virginialake}
\usepackage{url}
\usepackage{listliketab}
\usepackage{colortbl}
\usepackage{colonequals}
\usepackage[all]{xy}
\usepackage{tikz}
\usetikzlibrary{positioning,arrows,calc,shapes.geometric}
\tikzset{ modal/.style={>=stealth,shorten >=1pt,shorten <=1pt,auto,node distance=0.7cm, semithick}, world/.style={circle,draw,minimum size=0.2cm}, point/.style={circle,draw,inner sep=0.5mm,fill=black},point1/.style={circle,draw,inner sep=0.5mm}, reflexive above/.style={->,loop,looseness=7,in=120,out=60}, reflexive below/.style={->,loop,looseness=7,in=240,out=300}, reflexive left/.style={->,loop,looseness=7,in=150,out=210}, reflexive right/.style={->,loop,looseness=7,in=30,out=330}, itria/.style={draw,isosceles triangle,shape border rotate=270,yshift=-1.45cm, minimum height = 15mm}, empty/.style={inner sep=0.0mm } }

\renewcommand{\emptyset}{\varnothing}
\renewcommand{\phi}{\varphi}
\newcommand{\impl}{\rightarrow}
\renewcommand{\epsilon}{\varepsilon}
\renewcommand{\L}{{\sf L}}
\newcommand{\IPC}{{\sf IPC}}
\newcommand{\K}{{\sf K}}
\newcommand{\D}{{\sf D}}
\newcommand{\T}{{\sf T}}
\newcommand{\GL}{{\sf GL}}
\newcommand{\Kfour}{{\sf K4}}
\newcommand{\Sfour}{{\sf S4}}
\newcommand{\Sfive}{{\sf S5}}
\newcommand{\NK}{{\sf NK}}
\newcommand{\ND}{{\sf ND}}
\newcommand{\NT}{{\sf NT}}
\newcommand{\nestedK}{{\sf NK}}
\newcommand{\logic}{{\sf L}}
\newcommand{\imp}{\to}
\newcommand{\vlfill}[1]{\{#1\}}
\renewcommand{\vlhole}{\vlfill{\:}}

\newcommand{\idTnd}{\mathsf{id_P}}
\newcommand{\idTop}{\mathsf{id_{\top}}}

\newcommand{\ruled}{{\sf d}}
\newcommand{\rulet}{{\sf t}}

\newcommand{\ce}{\colonequals}
\newcommand{\cce}{\coloncolonequals}

\newcommand{\calM}{\mathcal{M}}
\newcommand{\calN}{\mathcal{N}}
\newcommand{\calI}{\mathcal{I}}
\newcommand{\calJ}{\mathcal{J}}

\newcommand{\calL}{\mathcal{L}}

\newcommand{\calG}{\mathcal{G}}
\newcommand{\calH}{\mathcal{H}}

\newcommand{\spdisj}{\mathbin{\varovee}}
\newcommand{\spconj}{\mathbin{\varowedge}}
\newcommand{\bigspdisj}{\mathop{\mathlarger{\mathlarger{\mathlarger{\mathlarger{\varovee}}}}}\limits}
\newcommand{\bigspconj}{\mathop{\mathlarger{\mathlarger{\mathlarger{\mathlarger{\varowedge}}}}}\limits}
\newcommand{\form}{{\sf form}}

\newcommand{\Var}{\mathit{Var}}

\newcommand{\Icomment}[1]{ }

\newbool{arxiv}
\booltrue{arxiv} 

\nolinenumbers

 \title{Uniform interpolation via nested sequents \ifarxiv and hypersequents\fi}
 
  \author{Iris van der Giessen}{Utrecht University, Netherlands \and \url{https://www.uu.nl/staff/IvanderGiessen/}}{i.vandergiessen@uu.nl}{}{}
  
  \author{Raheleh Jalali}{Utrecht University, Netherlands \and \url{https://www.uu.nl/medewerkers/RJalaliKeshavarz}}{r.jalalikeshavarz@uu.nl}{}{}
  
  \author{Roman Kuznets}{TU Wien, Austria \and \url{https://sites.google.com/site/kuznets/}}{roman@logic.at}{}{}
  
  \authorrunning{I. van der Giessen, R. Jalali, and R. Kuznets}

  \Copyright{2021 Iris van der Giessen, Raheleh Jalali, and Roman Kuznets}
  


\funding{Iris van der Giessen and Raheleh Jalali  acknowledge the support of the Netherlands Organization for Scientific Research under grant 639.073.807.
Roman Kuznets is funded by the Austrian Science Fund (FWF) ByzDEL project P 33600. }
  
\hideLIPIcs

\begin{document}

\maketitle

\begin{abstract}
A modular proof-theoretic framework was recently developed to prove Craig interpolation for normal modal logics based on generalizations of sequent calculi (e.g.,~nested sequents, hypersequents, and labelled sequents). In this paper, we turn to  uniform interpolation, which is stronger than Craig interpolation. We develop a constructive method for proving uniform interpolation via nested sequents and apply it to reprove the uniform interpolation property for normal modal logics~$\K$,~$\D$,~and~$\T$. 
\ifarxiv We then use the know-how developed for nested sequents to apply the same method to hypersequents and obtain the first direct proof of uniform interpolation for~$\Sfive$ via a cut-free sequent-like calculus. \fi
While our method is proof-theoretic, the definition of uniform interpolation for nested sequents and hypersequents also uses  semantic notions, including bisimulation modulo an atomic proposition.
\end{abstract}

\section{Introduction}
Uniform interpolation is stronger than Craig interpolation and provides a simulation of quantifiers in a logic. Similar to Craig interpolation, uniform interpolation is useful in computer science, for example, in quantifier elimination procedures~\cite{GabSchmiSzas08} or in knowledge representation to perform tasks such as forgetting irrelevant information in descriptive logics~\cite{Koopmann15PhD}. This shows the practical value of uniform interpolation. The goal of this paper is to expand the reach of proof-theoretic method of proving uniform interpolation.

A propositional (modal) logic $\L$ admits the Craig interpolation property (CIP) if for any formulas~$\phi$~and~$\psi$ such that $\vdash_\L \phi \impl \psi$,  there is an interpolant $\theta$ containing only atomic propositions that occur in both $\phi$ and $\psi$ such that $\vdash_\L \phi \impl \theta$ and $\vdash_\L \theta \impl \psi$.    One could say that the purpose of the interpolant is to state the reason $\psi$ follows from $\phi$ by using the common language of the two. Logic $\L$ has the uniform interpolation property (UIP) if for each formula~$\phi$ and each atomic proposition~$p$ there are uniform interpolants~$\exists p \phi$~and~$\forall p \phi$ containing only atomic propositions occurring in $\varphi$ except for~$p$ such that for all formulas~$\psi$ not containing~$p$:\looseness=-1
\[
\vdash_\L \phi \impl \psi \ \Leftrightarrow \ \vdash_\L \exists p \phi \impl \psi \ \ \quad\text{ and }\quad \ \
\vdash_\L \psi \impl \phi \ \Leftrightarrow \ \vdash_\L \psi \impl \forall p \phi.
\]
It is well known that this property is stronger than Craig interpolation. Indeed, by computing uniform interpolants consecutively, it is possible to remove a given set of  atomic propositions and construct a formula that would uniformly serve as a Craig interpolant for a fixed $\phi$ and  all $\psi$ with a given common language. 

Analytic sequent calculi can be used to prove the CIP constructively. For the~UIP, terminating cut-free sequent calculi play a similar role. Whereas for the CIP the syntactic proofs are often straightforward, the case of the UIP is much more complicated. Pitts provided a first syntactic proof of this kind, establishing the UIP for $\IPC$~\cite{Pitts92JSL}. B\'{i}lkov\'{a} successfully adjusted the method to (re)prove the UIP for several modal logics including~$\K, \T$, and $\GL$~\cite{Bilkova06PhD}. Iemhoff  provided a modular method for (intuitionistic) modal logics and intermediate logics with sequent calculi consisting of so-called focused rules, among others  establishing the UIP for $\D$~\cite{Iemhoff19a, Iemhoff19b}.

There are also algebraic and model-theoretic methods. The UIP for~$\GL$ and $\K$ is due to Shavrukov~\cite{Shavrukov93}  and Ghilardi~\cite{Ghilardi95} respectively. Interestingly, modal logics~$\Sfour$~and~$\Kfour$ do not enjoy the UIP~\cite{Bilkova06PhD,GhiZaw95} despite enjoying the CIP. Visser provided purely semantic proofs for~$\K$,~$\GL$,~and~$\IPC$ based on bounded bisimulation up to atomic propositions~\cite{Visser96}.  This method was later applied to prove the stronger Lyndon UIP for a wide range of modal logics~\cite{Kur20}. The semantic interpretation of uniform interpolation is called bisimulation quantifiers, see~\cite{dAgostino07} for an extended explanation. Bisimulations will also play a role in the current paper. \looseness=-1

The proof-theoretic approach has two advantages. First, it enables one to find  interpolants constructively rather than merely prove their existence.\footnote{More precisely, the method enables one to find interpolants efficiently rather than by an exhaustive search of all formulas, the search that terminates due to the proven existence of an interpolant.} Second, it can turn  uniform interpolation into a powerful tool in the study of existence of proof systems. Negative results are obtained in \cite{Iemhoff19a, Iemhoff19b} stating that logics without the UIP cannot have certain natural sequent calculi. As a consequence, $\Kfour$ and $\Sfour$ do not possess such proof systems. Similar negative results are obtained for modal and substructural logics in~\cite{semianalytic1}~and~\cite{semianalytic2} using the~CIP and UIP. These methods are based on calculi with regular sequents. 

The goal of this paper is to extend the same line of research to multisequent formalisms such as hypersequents and nested sequents\footnote{Nested sequents are also known as tree-hypersequents~\cite{Pog09TL} or deep sequents~\cite{Brue09AML} in the literature.}. Such forms of sequent calculi have recently been adapted to prove the CIP of modal logics via nested sequents~\cite{FitKuz15APAL} and hypersequents~\cite{Kuz16Proof}. A modular proof-theoretic framework encompassing these and also labelled sequents was provided in~\cite{Kuz18APAL}. The same ideas were extended to multisequent calculi for intermediate logics~\cite{KuzLell18}. The method combines syntactic and semantic reasoning. Generalized Craig interpolants are constructed using the calculus in a purely syntactic manner, but the method's correctness uses semantic notions from Kripke models of the underlying logic.

This paper extends this method providing proof-theoretic proofs for the UIP based on nested sequents for $\K$,~$\D$,~and~$\T$ \ifarxiv and on hypersequents for~$\Sfive$\fi.  The UIP for these logics has been known, but we provide a new method that can hopefully be extended to other logics. Similar to \cite{Kuz18APAL}, we combine syntactic and semantic reasoning. We use the terminating calculi to define the uniform interpolants and then provide model modifications and use bisimulations to prove the correctness of these interpolants.

B\'{i}lkov\'{a} \cite{Bilkova11} also provided a syntactic method for uniform interpolation for $\K$ based on nested sequents. She presented proofs based on two nested calculi for $\K$: one with a standard modality and another that is based on a different modal language with a cover modality~$\nabla$. B\'{i}lkov\'{a}'s method for nested sequents is closely related to her work based on regular sequents in~\cite{Bilkova06PhD}. The main difference with our method is that we exploit the treelike structure of nested sequents reflecting the treelike models for $\K$ by incorporating  semantic arguments while the algorithm for the interpolant computation remains fully syntactic.  We intend  our method to form a good basis for generalizing to other logics with multisequent calculi.

The paper is organized as follows. In Sect.~\ref{sec:preliminaries}, we introduce the nested sequent calculi for~$\K$,~ $\T$,~and~$\D$, as well as model modifications invariant under bisimulation. In Sect.~\ref{sec:uip}, we present our method to prove uniform interpolation for~$\K$,~$\T$,~and~$\D$. \ifarxiv  In Sect.~\ref{sec:uipS5}, we show how the method can be adjusted to hypersequents for~$\Sfive$. \fi Section~\ref{sect:concl} concludes the paper and maps the immediate next steps.

\section{Preliminaries}\label{sec:preliminaries}
\begin{definition}
\emph{Modal formulas} in \emph{negation normal form} are defined by the following grammar $
\varphi \cce \bot \mid \top \mid p \mid \overline{p} \mid (\varphi \land \varphi) \mid (\varphi \lor \varphi) \mid \Box \varphi \mid \Diamond \varphi$
where\/ $\bot$~and\/~$\top$~are \emph{Boolean constants}, $p$~is an \emph{atomic proposition} from a countable set\/ {\sf Prop}, and $\overline{p}$~is the negation of~$p$ for each $p \in {\sf Prop}$. The set\/~{\sf Lit}  of \emph{literals} consists of all atomic propositions and their negations, with $\ell$~used to denote its elements. Literals and Boolean constants are \emph{atomic formulas}.
\end{definition}

We define~$\overline{\phi}$ (or~$\neg \phi$) recursively as usual using De~Morgan's laws to push the negation inwards. $\phi \to \psi\ce\overline{\phi} \lor \psi$.

\begin{definition}
\label{def:NestedSequents}
\emph{Nested sequents~$\Gamma$} are recursively defined in the following form: 
\[
\phi_1, \dots, \phi_n, [\Gamma_1], \dots, [\Gamma_m]
\]
is a nested sequent where $\phi_1, \dots, \phi_n$ are modal formulas for $n \geq 0$ and\/ $\Gamma_1,\dots,\Gamma_m$ are nested sequents for $m \geq 0$.
We call brackets\/~$[\enspace]$ a \emph{structural box}. The \emph{formula interpretation~$\iota$} of a nested sequent is defined recursively~by 
\[
\iota(\phi_1, \dots, \phi_n, [\Gamma_1], \dots, [\Gamma_m]) \ce \phi_1 \lor \dots \lor \phi_n \lor \Box\iota(\Gamma_1) \lor \dots \lor \Box \iota(\Gamma_m).
\]
\end{definition}

One way of looking at a nested sequent is to consider a tree of ordinary (one-sided) sequents, i.e.,~of multisets of formulas. Each structural box in the nested sequent creates a child in the tree. In order to be able to reason about formulas in a particular tree node, we introduce labels. A \emph{label} is a finite sequence of natural numbers. We denote labels by~$\sigma, \tau, \dots$; a label~$\sigma \ast n$ (or simply~$\sigma n$) denotes the label~$\sigma$ extended by the natural number~$n$.

\begin{definition}[Labeling]
For a nested sequent\/~$\Gamma$ and label~$\sigma$ we define a \emph{labeling function~$l_\sigma$} to recursively label structural boxes in nested sequents as follows: 
\[
l_\sigma(\phi_1, \dots, \phi_n, [\Gamma_1], \dots, [\Gamma_m]) \ce \phi_1, \dots, \phi_n, [l_{\sigma*1}(\Gamma_1)]_{\sigma*1}, \dots, [l_{\sigma*m}(\Gamma_m)]_{\sigma*m}.
\]
Let $\calL_\sigma(\Gamma)$~be the set of labels occurring in~$l_\sigma(\Gamma)$ plus label~$\sigma$ (for formulas outside all structural boxes). Define the labeled nested sequent $l(\Gamma) \ce l_1(\Gamma)$, and let $\calL(\Gamma)\ce\calL_1(\Gamma)$.\footnote{Labeled nested sequents are closely related to labelled sequents from~\cite{NegvPla11} but retain the nested notation.}

Formulas in a nested sequent\/~$\Gamma$ are \emph{labeled} according to the labeling of the structural boxes containing them. We write\/ $1 : \phi \in \Gamma$ if{f} the formula~$\phi$ occurs in\/~$\Gamma$ outside all structural boxes. Otherwise, $\sigma : \phi \in \Gamma$ whenever $\phi$ occurs in~$l(\Gamma)$ within a structural box labeled~$\sigma$. 
\end{definition}

The set~$\calL(\Gamma)$ can be considered as the set of nodes of the corresponding tree of~$\Gamma$, with $1$~being the root of this tree. Often, we do not distinguish between a nested sequent~$\Gamma$ and its labeled sequent~$l(\Gamma)$. For example, we say that $\sigma \in \Gamma$ if $\sigma \in \calL(\Gamma)$. 

\ifarxiv 
\begin{example}\label{ex:nestedsequent}
Consider a nested sequent $\Gamma =  \varphi, [p, \psi ], \bigl[\overline{p}, \phi,  [\chi]\bigr]$. The corresponding labeled nested sequent is 
$
l(\Gamma) = \varphi, [p, \psi ]_{11}, \left[\overline{p}, \phi,  [\chi]_{121}\right]_{12}
$
with $\calL(\Gamma)=\{1, 11, 12, 121\}$.
The corresponding tree is pictured as follows, where each node is labeled on the left and marked by its formulas on the right (in particular, here $1 : \phi \in \Gamma$ and $121: \chi \in \Gamma$, but  $12: \chi \notin \Gamma$):
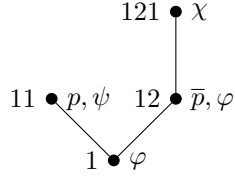
\begin{center}
\begin{tikzpicture}[node distance=1cm,sibling distance=10em,
  point/.style={circle,draw,inner sep=0.5mm,fill=black},]]
  \node[point] (w) [label=left:{1}, label=right:{$\varphi$}] {};
  \node[point] (v1) [above left= of w, label=left:{11}, label=right:{$p, \psi$}] {};
  \node[point] (v2) [above right= of w, label=left:{12}, label=right:{$\overline{p}, \phi$ }] {};
  \node[point] (x) [above = of v2, label=left:{121}, label=right:{$\chi$}] {};
\path[-] (w) edge (v1);
\path[-] (w) edge (v2);
\path[-] (v2) edge (x);
\end{tikzpicture}
\end{center}
\end{example}

Following~\cite{Brue09AML}, we will work with contexts in rules to signify that the rules can be applied in an arbitrary node of the nested sequent. We will work with \emph{unary contexts} which are nested sequents with exactly one \emph{hole}, denoted by the symbol~$\vlhole$. Such contexts are denoted by~$\Gamma \vlhole$. The insertion~$\Gamma \{ \Delta \}$ of a nested sequent~$\Delta$ into a context~$\Gamma\vlhole$ is obtained by replacing the occurrence~$\vlhole$ with~$\Delta$. The hole~$\vlhole$ can be labeled the same way as  formulas. We write~$\Gamma \vlhole_\sigma$ to denote the label of the hole.

\begin{example}
$\Gamma'\vlhole =  \varphi, [p, \psi ], [\overline{p}, \vlhole]$ is  a context. Its labeled context is $\Gamma'\vlhole_{12} =  \varphi, [p, \psi ]_{11}, [\overline{p}, \vlhole]_{12}$. Let $\Delta = \phi, [\chi]$. Then $\Gamma'\{ \Delta \}$~equals~$\Gamma$ from Example \ref{ex:nestedsequent}. 
\end{example}
\fi

\begin{definition}[Variables]
Whether $X$~is a formula, or a sequence/set/multiset of formulas, or a nested sequent/context, or some other formula-based object, we denote by $\Var(X) \subseteq {\sf Prop}$ the set of atomic propositions occurring in~$X$ (note that $p$~may also occur in the form of~$\overline{p}$).
\end{definition}

In this paper we use nested sequent calculi for classical modal logics~$\K$, $\D$,~and~$\T$ from~\cite{Brue09AML}. Recall that $\K$~consists of all classical tautologies, the ${\sf k}$-axiom $\Box(\phi \impl \psi) \impl (\Box \phi \impl \Box \psi)$ and is closed under \emph{modus ponens} (from $\phi \impl \psi$ and~$\phi$, infer~$\psi$) and \emph{necessitation} (from~$\phi$, infer~$\Box \phi$).  
Further,  $\D \ce \K + \Box \phi \impl \Diamond \phi$ and $\T \ce \K + \Box \phi \impl \phi$. 
We now introduce  nested sequent calculi and then Kripke semantics for these logics.

The terminating nested sequent calculus~\nestedK{} for the modal logic~\K{} consists of all rules in the first two rows in  Fig.~\ref{fig:rules2} plus the rule~{\sf k}. 
This calculus is an extension of the multiset-based version from~\cite{Brue09AML}
to the language with  Boolean constants~$\bot$~and~$\top$, necessitating an addition 
of the rule~$\idTop$ for handling these (cf. also the treatment of Boolean constants in~\cite{FitKuz15APAL}). 
The nested calculus~{\sf ND}~({\sf NT}) for the modal logic~{\sf D}~({\sf T}) is obtained by adding to~\nestedK{} the rule~$\ruled$~($\rulet$). As shown in~\cite{Brue09AML}, the nested sequent calculi~$\NK$, $\ND$,~and~$\NT$ are sound and complete for modal logics~$\K$, $\D$,~and~$\T$ respectively, i.e.,~a nested sequent~$\Gamma$ is derivable in~$\NK$~($\ND$,~$\NT$) if and only if its formula interpretation~$\iota(\Gamma)$ is a theorem of~$\K$~($\D$,~$\T$).
\begin{figure}[h]
\centering
  \fbox{\parbox{.7\textwidth}{
    \begin{center}
      $\vlinf{\idTnd}{}{\Gamma\vlfill{p,\overline p}}{}$
      \qquad
      $\vlinf{\idTop}{}{\Gamma\vlfill{\top}}{}$
      \qquad
      $\vlinf{\lor}{}{\Gamma\vlfill{\varphi \lor \psi}}{\Gamma\vlfill{\varphi \lor \psi,  \varphi,\psi}}$
      \end{center}  
      \begin{center}
      \quad
      $\vliinf{\land}{}{\Gamma\vlfill{\varphi \land \psi}}{\Gamma\vlfill{\varphi \land \psi,\varphi}}{\Gamma\vlfill{\varphi \land \psi,\psi}}$
      \qquad
      $\vlinf{\Box}{}{\Gamma\vlfill{\Box \varphi}}{\Gamma\vlfill{\Box \varphi,[\varphi]}}$
      \end{center}
        \begin{center}
      $\vlinf{\mathsf{k}}{}{\Gamma\vlfill{\Diamond \varphi,[\Delta]}}{\Gamma\vlfill{ \Diamond \varphi,[\Delta,\varphi]}}$
    \qquad
      $\vlinf{\ruled}{}{\Gamma\vlfill{\Diamond \varphi}}{\Gamma\vlfill{\Diamond \varphi,[\varphi]}}$
      \qquad
      $\vlinf{\rulet}{}{\Gamma\vlfill{\Diamond \varphi}}{\Gamma\vlfill{\Diamond \varphi,\varphi}}$
    \end{center}
  }}
\caption{Terminating nested rules: the principal formula is not saturated.}
\label{fig:rules2}
\end{figure}

\begin{definition}[Saturation] \label{def:saturated}
Consider a sequent\/ $\Gamma=\Gamma'\{\theta \}_\sigma$, i.e.,~$\sigma : \theta \in \Gamma$. The formula~$\theta$ is \emph{$\K$-saturated in~$\Gamma$} if the following conditions hold depending on the form of~$\theta$:\looseness=-1
\begin{itemize}
\item $\theta$~is  an atomic formula;
\item if $\theta = \varphi \lor \psi$, then both $\sigma : \varphi\in \Gamma$ and $\sigma : \psi\in \Gamma$;
\item if $\theta = \varphi \land \psi$, then  either $\sigma : \varphi\in \Gamma$ or $\sigma : \psi\in \Gamma$; 
\item if $\theta = \Box \varphi$, then there is a label $\sigma \ast n \in \calL(\Gamma)$ such that $\sigma \ast n : \varphi\in \Gamma$.
\end{itemize}
The formula~$\theta$ of the form~$\Diamond \varphi$ is 
\begin{itemize}
\item \emph{$\K$-saturated in~$\Gamma$  w.r.t.~$\sigma\ast n\in \calL(\Gamma)$} if
$\sigma \ast n : \varphi\in \Gamma$;
\item
\emph{$\D$-saturated in~$\Gamma$} if 
 there is a label $\sigma \ast n  \in \calL(\Gamma)$;
\item
\emph{$\T$-saturated in~$\Gamma$} if 
$\sigma : \varphi\in \Gamma$.
\end{itemize}
A nested sequent\/~$\Gamma$ is\/ $\K$-\emph{saturated}  if\/ $(1)$~it is neither of the form\/ $\Gamma'\{p,\overline{p} \}$ for some atomic proposition $p \in {\sf Prop}$ nor of the form\/ $\Gamma'\{\top\}$;\/ $(2)$~all its formulas $\sigma:\Diamond \varphi$ are\/ $\K$-saturated w.r.t.~every child of~$\sigma$; and\/ $(3)$~all its other formulas are\/  $\K$-saturated in\/~$\Gamma$. A nested sequent is\/ $\D$-saturated\/~$(\T$-saturated$)$ if it is\/ $\K$-saturated and all its formulas $\sigma:\Diamond \varphi$ are\/ $\D$-saturated\/~$(\T$-saturated$)$ in\/~$\Gamma$. \looseness=-1
\end{definition}

\ifarxiv 
\begin{example} 
The sequent~$\Gamma=[\Diamond \varphi]$ is $\K$-saturated  but neither $\D$-saturated  nor $\T$-saturated. Indeed, for the logic~$\D$ we would need~$1*1*n:\varphi$ to be present for some~$n$ and for the logic~$\T$ we would need to have~$1*1:\varphi$ in order to saturate $1*1 : \Diamond \varphi\in\Gamma$.
\end{example}

The rules from Fig.~\ref{fig:rules2} with embedded contraction are sometimes called Kleene'd rules. Following~\cite{Brue09AML}, in order to ensure finite proof search, we only apply a rule when the principal formula in the conclusion is not saturated w.r.t.~this rule, i.e.,~$\varphi \lor \psi$, $\varphi \land \psi$, and~$\Box \varphi$ are not $\K$-saturated, $\Diamond \varphi$~in the rule~{\sf k} is not $\K$-saturated w.r.t.~the label of the bracket containing~$\Delta$,  $\Diamond \varphi$~in the rule~{\sf d} is not $\D$-saturated, and  $\Diamond \varphi$~in the rule~{\sf t} is not $\T$-saturated. Since for Kleene'd rules principal formulas are preserved in the premises, the number of applications of each of the rules~$\mathsf{k}$, $\mathsf{d}$,~and~$\mathsf{t}$ is bounded. The way to think of a saturated sequent is that in a bottom-up proof search when we reach a saturated sequent, it does not make sense to apply more rules as these would only lead to duplications. 
\fi

\begin{theorem}[Br\"{u}nnler~\cite{Brue09AML}]\label{lem:terminationK}
The calculi for\/~$\K, \D$,~and\/~$\T$ in Fig.~\ref{fig:rules2} are terminating.
\end{theorem}

Intuitively, nested sequents capture the tree structure of Kripke models for modal logics. 
We define truth for nested sequents in Kripke models and then  recall relevant facts about bisimulations and introduce model modifications that we use in the proof of uniform interpolation.\looseness=-1

\begin{definition} \label{def: Kripke Frames and Models}
A \emph{Kripke model} is a triple $\calM = (W,R,V)$, where 
$W\ne \emptyset$\ifarxiv{} is a  set of \emph{worlds} or \emph{nodes}\fi, 
$R\subseteq W \times W$, and 
$V : {\sf Prop} \to 2^W$ is a \emph{valuation function}\ifarxiv{} mapping each atomic proposition $p \in {\sf Prop}$ to a set~$V(p)$ of worlds from~$W$\fi. 
\ifarxiv{}If~$vRw$, we say that $w$~is \emph{accessible from}~$v$, or that $v$~is a \emph{parent} of~$w$, or that $w$~is a \emph{child} of~$v$. \fi
We define  $\calM, w \models \phi$ by induction on the construction of~$\phi$ as usual:
 $\calM, w \models \top$ and $\calM, w \not \models \bot$; 
 for $p \in {\sf Prop}$, we have $\calM, w \models p$ if{f}  $w \in V(p)$ and $\calM, w \models \overline{p}$ if{f} $w \notin V(p)$;  
 we have $\calM, w \models \varphi \wedge \psi$ $(\calM, w \models \varphi \vee \psi)$ if{f} $\calM, w \models \varphi$ and (or) $\calM, w \models \psi$; 
 finally, $\calM, w \models \Box \varphi$ if{f} $\calM, v \models \varphi$ whenever~$wRv$ and 
 $\calM, w \models \Diamond \varphi$ if{f} $\calM, v \models \varphi$ for some~$wRv$.
A formula~$\phi$ is \emph{valid in~$\calM$}, denoted $\calM \models \phi$, when $\calM, w \models \phi$ for all $w \in W$.  

A model $\calM'=(W',R',V')$ is a \emph{submodel} of $\calM=(W,R,V)$ when $W'\subseteq W$, $R' = R \cap (W' \times W')$, and $V'(p) = V(p)\cap W'$ for each $p \in {\sf Prop}$. A \emph{submodel generated by $w \in W$}, denoted~$\calM_w=(W_w,R_w,V_w)$, is the smallest submodel $\calM'=(W',R',V')$ of $\calM$  such that $(1)$~$w \in W'$ and $(2)$~$v \in W'$ whenever~$xRv$ and $x \in W'$. 
\end{definition}

We will use models based on  finite intransitive directed trees,  usually denoting the \emph{root}~$\rho$. For~$\T$, the accessibility relation~$R$ is required to be reflexive, i.e.,~$\forall w \in W w R w$. For~$\D$, the accessibility relation~$R$ must be serial, i.e.,~$\forall w \in W \exists v \in W w R v$. Note that such seriality implies reflexivity of the leaves of the tree. Finally, we assume~$R$ to be irreflexive  for~$\K$. From now on we call these models $\T$-models, $\D$-models, and $\K$-models respectively. 

\begin{theorem}
If\/ $\L \in \{\K, \D, \T\}$, then\/
$\vdash_\L \phi$ if{f} $\calM \models \phi$ for each\/ $\L$-model~$\calM$.
\end{theorem} 

Following~\cite{Kuz18APAL}, we now extend the definitions of truth and validity to nested sequents.

\begin{definition}\label{def: Multiworld Interpretation2}\label{def:truth_seq}
A \emph{(treelike) multiworld interpretation of a nested sequent\/~$\Gamma$ into a model $\calM=(W,R,V)$} is a function $\calI : \calL(\Gamma) \impl W$ from labels in\/~$\Gamma$ to worlds of~$\calM$ such that $\calI(\sigma)R\calI(\sigma \ast n)$ whenever $\{\sigma, \sigma \ast n\} \subseteq \calL(\Gamma)$. 
Then 
\[
\calM, \calI \models \Gamma 
\qquad \Longleftrightarrow \qquad
\calM, \calI(\sigma) \models \varphi \text{ for some }  \sigma : \varphi \in \Gamma.
\]
$\Gamma$~is \emph{valid in~$\calM$}, denoted $\calM \models \Gamma$, means that $\calM, \calI \models \Gamma$ for all multiworld interpretations~$\calI$ of\/~$\Gamma$ into~$\calM$. 
\end{definition}

The following lemma, which can be easily proved by induction on the structure of~$\Gamma$, implies completeness for validity of nested sequents.

\begin{lemma}\label{lem:completenessnested}
For a nested sequent\/~$\Gamma$  and a model~$\calM$, we have $\calM \models \Gamma$ if{f} $\calM \models \iota(\Gamma).$
\end{lemma}
\ifarxiv\begin{proof}
By induction on the structure of~$\Gamma$, 
we prove  that $\calM, \calI \not \models \Gamma$ implies $\calM,\calI(1) \not \models \iota(\Gamma)$ for one direction and that  $\calM, w \not \models \iota(\Gamma)$ implies $\calM,\calI \not \models \Gamma$ for some~$\calI$ such that $\calI(1)=w$ for the other direction. Let $\Gamma$~be of the form $\phi_1, \dots, \phi_n, [\Gamma_1], \dots, [\Gamma_m]$. 

First suppose $\calM, \calI \not \models \Gamma$. Then for  all $\sigma:\psi \in \Gamma$ we have $\calM, \calI(\sigma) \not \models \psi$, in particular,  $\calM, \calI(1) \not \models \phi_i$ for all~$i$. In addition, we show that $\calM, \calI(1) \not \models \Box \iota(\Gamma_j)$ for all~$j$. To prove this, we define~$\calI_j$ as follows: $\calI_j(1 * \sigma') \ce \calI(1*j*\sigma')$ for each $1 * \sigma' \in \calL(\Gamma_j)$; in particular, $\calI_j(1) \ce \calI(1*j)$. 
It is easy to see that $\calI_j$~is a multiworld interpretation of~$\Gamma_j$ into~$\calM$  and that  $\calM, \calI_j \not \models \Gamma_j$. Thus, by induction hypothesis,  $\calM, \calI_j(1) \not \models \iota(\Gamma_j)$, i.e.,~$\calM, \calI(1*j) \not \models \iota(\Gamma_j)$. Since $\calI(1)R\calI(1*j)$, it follows that $\calM, \calI(1) \not \models \Box \iota(\Gamma_j)$. We conclude that $\calM, \calI(1) \not \models \iota(\Gamma)$.

Now suppose $\calM, w \not \models \iota(\Gamma)$. For each~$j$, there is a world~$v_j$ such that~$wRv_j$ and $\calM, v_j \not \models \iota(\Gamma_j)$. By induction hypothesis, there exists a multiworld interpretation~$\calI_j$ of~$\Gamma_j$ into~$\calM$ such that $\calI_j(1) = v_j$ and $\calM, \calI_j \not \models \Gamma_j$. Define~$\calI$ as follows: $\calI(1)\ce w$ and $\calI(1*j*\sigma) \ce \calI_j(1*\sigma)$. We immediately have $\calM, \calI \not \models \Gamma$.
\end{proof}\fi

We now define bisimulations modulo an atomic proposition~$p$, similar to the ones from~\cite{dAgostino07,Visser96}, where uniform interpolation is studied on the basis of bisimulation quantifiers. While those papers focus on purely semantic methods, we embed the semantic tool of bisimulation into our constructive proof-theoretic approach in Sect.~\ref{sec:uip}.
Our bisimulations  behave largely like standard bisimulations except they do not have to preserve the truth of formulas with occurrences of~$p$. 
\begin{definition}[Bisimilarity] \label{def: Bisimilarity} 
A \emph{bisimulation  up to an atomic proposition~$p$ between models $\calM=(W,R,V)$ and $\calM'=(W',R',V')$} is a non-empty binary relation $Z \subseteq W \times W'$ such that the following conditions hold for all $w \in W$ and $w' \in W'$ with~$wZ w'$:
\begin{description}
\item[\textup{atoms}$_p$.] 
$w \in V(q)$ if{f} $w' \in V'(q)$ for all $q \in {\sf Prop} \setminus \{p\}$;
\item[\textup{forth}.]
if~$w R v$, then there exists $v'\in W'$ such that $v Z v'$~and~$w' R' v'$; and
\item[\textup{back}.]
 if~$w' R' v'$, then there exists $v \in W$ such that $v Z v'$~and~$w R v$.
\end{description}
When~$w Z w'$, we write $(\calM, w) \sim_p (\calM', w')$. Further, we write
 $(\calM, \calI) \sim_p (\calM',\calI')$ for functions $\calI : X \to W$ and $\calI' : X \to W'$ with a common domain~$X$ if there is a bisimulation~$Z$ up to~$p$ between~$\calM$~and~$\calM'$  such that $ \calI(\sigma) Z \calI'(\sigma)$ for each $\sigma \in X$.
\end{definition} 

The main property of bisimulations is truth preservation for modal formulas. The following theorem is proved the same way as  \cite[Theorem 2.20]{blackburn}.

\begin{theorem}\label{thm:invariantform}
If $(\calM,w) \sim_p (\calM',w')$, then for all formulas~$\phi$ with $p \notin \Var(\phi) $, we have $\calM,w \models \phi$ if{f} $\calM',w' \models \phi$. 
\end{theorem}

We are interested in manipulations of  treelike models that preserve bisimulations up to~$p$, in particular, in duplicating a part of a model or replacing it with a bisimilar model.

\begin{definition}[Model transformations]\label{def:copymodels}
Let $\calM = (W,R,V)$ be an intransitive tree (possibly with some reflexive worlds), $\calM_w = (W_w, R_w, V_w)$ be its subtree with root~$w \in W$, and  $\calN=(W_N,R_N,V_N)$ be another tree  with root $\rho_N \in W_N$. 
A model $\calM'\ce(W', R', V')$ is the result of \emph{replacing the subtree~$\calM_w$  with $\calN$ in $\calM$} if 
\ifarxiv 
\begin{align*}
W' &\ce (W \setminus W_w) \sqcup W_N,
\\  
R'&\ce \bigl(R \cap (W\setminus W_w)^2 \bigr)\sqcup R_N \sqcup \bigl\{(v,\rho_N) \mid vRw\bigr\},
\\
V'(q) &\ce \bigl(V(q) \setminus W_w\bigr) \sqcup V_N(q) \text{ for all  $q\in {\sf Prop}$}.
\end{align*}
\else 
$W' \ce (W \setminus W_w) \sqcup W_N$,  $V'(q) \ce (V(q) \setminus W_w) \sqcup V_N(q)$ for all  $q\in {\sf Prop}$, and $
R'\ce (R \cap (W\setminus W_w)^2 )\sqcup R_N \sqcup \{(v,\rho_N) \mid vRw\}$.
\fi

A model $\calM''\ce(W'',R'',V'')$ is the result of \emph{duplicating} $($\emph{cloning}$)$ $\calM_w$ in $\calM$  if another copy\/\footnote{Here $v^c\ce(v,c)$,  $W_w^c \ce \{v^c \mid v \in W_w\}$, $R_w^c \ce \{(v^c, u^c) \mid (v,u) \in R_w\}$, and $V_w^c(q) \ce \{v^c\mid v \in V_w(q)\}$.} $\calM^c_w\ce(W_w^c,R^c_w,V^c_w)$ of $\calM_w$ is inserted alongside\/  $($as a subtree of\/$)$~$\calM_w$, i.e., if 
\ifarxiv 
\begin{gather*}
W'' \ce W \sqcup W^c_w,
\\
R'' \ce R \sqcup  R^c_w  \sqcup \bigl\{(v,w^c) \mid vRw\bigr\} 
\text{$($duplicating\/$)$}
\quad\text{or}\quad
R'' \ce R \sqcup  R^c_w  \sqcup \bigl\{(w,w^c)\bigr\}
\text{$($cloning\/$)$},
\\
V''(q) \ce V(q) \sqcup V^c_w(q) \text{ for all $q\in {\sf Prop}$}.
\end{gather*}
\else
$W'' \ce W \sqcup W^c_w $, $V''(q) \ce V(q) \sqcup V^c_w(q) $ for all $q\in {\sf Prop}$, 
and, in case of duplicating, $R'' \ce R \sqcup  R^c_w  \sqcup \{(v,w^c) \mid vRw\}$  $($in case of cloning, $R'' \ce R \sqcup  R^c_w  \sqcup \{(w,w^c)\})$. \fi
\end{definition}

\begin{lemma}
\label{lem: Bisimilarity of Submodels} 
In the setup from Def.~\ref{def:copymodels}, let  $Z \subseteq W_N \times W_w$ be a bisimulation demonstrating that\/ $(\calN,\rho_N) \sim_p (\calM_w,w)$.  
Then, for~$\calM'$ obtained  by replacing~$\calM_w$ with~$\calN$ in~$\calM$ we have that\/ $(\calM', v) \sim_p (\calM,v)$ for all $v \in W\setminus W_w$ and  that\/ $(\calM',u_N) \sim_p (\calM,u)$ whenever~$u_N Z u$. Moreover, if both~$\calM$~and~$\calN$ are\/ $\K$-models\/~$(\D$-models,\/~$\T$-models\/$)$, then so is~$\calM'$.

For~$\calM''$ obtained by duplicating~$\calM_w$ in~$\calM$, we have\/ $(\calM'',v) \sim_p (\calM,v)$ for all $v \in W$ and, in addition,\/ $(\calM'', u^c) \sim_p (\calM, u)$ for all $u \in W_w$.  If $\calM$~is  a\/ $\K$-model~$(\D$-model,\/~$\T$-model\/$)$  not rooted at~$w$, so is~$\calM''$. 

The same holds for cloning if~$w R w$, except that cloning does not preserve\/ $\D$-models.
\end{lemma}

\begin{proof}
It is easy to see that one bisimulation witnesses all the stated bisimilarities in each case:
$Z' \ce \{(v,v) \mid v \in W \setminus W_w \} \sqcup Z$ for replacing or  
$Z'' \ce \{(v,v) \mid v \in W \} \sqcup \{(u^c,u) \mid u \in W_w\}$ for duplicating and cloning. Both the tree structure and reflexivity of worlds are preserved by all operations. Leaves are preserved by replacement and duplication, whereas cloning turns a leaf~$w$ into a non-leaf without removing its reflexivity as required in $\D$-models.
\end{proof}

\section{Uniform interpolation for nested sequents}\label{sec:uip}
In this section we prove the uniform interpolation theorem for~$\K$, $\T$,~and~$\D$ via their  nested sequent calculi~$\NK, \NT$,~and~$\ND$ respectively. We define  a new notion of uniform interpolation for nested sequents in Def.~\ref{def:Nested uip 2} that involves  Kripke semantics. We then prove in  Lemma~\ref{lem:uip equiv} that this implies  the standard definition of uniform interpolation.

\begin{definition}[Uniform interpolation property] \label{def:UniformInterpolation}
A logic\/~$\logic$  in a language containing an implication\/~$\imp$ and Boolean constants\/~$\bot$~and\/~$\top$ (primary or defined) has the \emph{uniform interpolation property}, or~\emph{UIP}, if for every formula~$\varphi$ in the logic and atomic proposition~$p$, there exist formulas~$\forall p \varphi$~and~$\exists p \varphi$ 
such that 
\begin{enumerate}[(i)]
\item\label{UIP:1} $ \Var(\exists p \varphi) \subseteq  \Var(\varphi) \setminus \{ p \}$ and $ \Var(\forall p \varphi) \subseteq  \Var(\varphi) \setminus \{ p \}$,
\item\label{UIP:2} $\vdash_\logic \varphi \imp \exists p \varphi \text{ and } \vdash_\logic \forall p \varphi \imp  \varphi,$ and
\item\label{UIP:3} for each formula~$\psi$ with  $p \notin \Var(\psi)$: \[\vdash_\logic \varphi \imp \psi \ \Rightarrow \ \vdash_\logic \exists p \varphi \imp \psi \qquad \vdash_\logic \psi \imp \varphi \ \Rightarrow \ \vdash_\logic \psi \imp \forall p \varphi.\]
\end{enumerate}
\end{definition}

For classical-based  logics, the existence of left-interpolants ensures the existence of right-interpolants, and vice versa. Assuming $\forall p \varphi$~is defined for each formula~$\varphi$, we can define $\exists p \varphi := \neg\forall p \overline{\varphi}$. Thus, from now on, we focus on~$\forall p \phi$.

We import some of the  notation from~\cite{Kuz18APAL} in order  to formulate the uniform interpolation property for nested sequents.

\begin{definition}
\label{def:multiformula}
\emph{Multiformulas} are defined by the grammar
\ifarxiv 
\[
\mho \cce \sigma : \varphi \mid (\mho \spconj \mho) \mid (\mho \spdisj \mho),
\]
\else
$\mho \cce \sigma : \varphi \mid (\mho \spconj \mho) \mid (\mho \spdisj \mho)$,
\fi
where $\sigma$~is a label and $\varphi$~is a formula. We write~$\calL(\mho)$ for the set of labels occurring in~$\mho$.
\end{definition}

\ifarxiv
\begin{remark} 
The symbol~$\mho$ is pronounced  `mho', which is the reverse of `ohm' the same way as $\mho$~is the reverse of~$\Omega$, the symbol for ohm in physics.
\end{remark}
\fi

\begin{definition}[Suitability]\label{def:interpretationmultiformula}
A multiworld interpretation~$\calI$ of a sequent\/~$\Gamma$ is \emph{suitable for a multiformula\/~$\mho$} if  $\calL(\mho) \subseteq \calL(\Gamma)$, in which case we call it  a \emph{multiworld interpretation of~$\mho$ into~$\calM$}.\looseness=-1
\end{definition}

\begin{definition}[Truth for multiformulas]
\label{def:truth_multif}
Let $\calI$~be a multiworld interpretation of a multiformula\/~$\mho$ into a model~$\calM $. We define $\calM, \calI \models \mho$ recursively as follows:

\begin{listliketab} 
    \storestyleof{enumerate} 
        \begin{tabular}{lll}
             &$\calM,\calI \models \sigma : \varphi$& if{f} \; $\calM, \calI(\sigma) \models \varphi$,\\
             &$\calM, \calI \models \mho_1 \spconj \mho_2$ & if{f} \; $\calM, \calI \models \mho_i$ for both $i=1,2$,\\
             &$\calM, \calI \models \mho_1 \spdisj \mho_2$& if{f} \; $\calM, \calI \models \mho_i$ for at least one $i=1,2$.
        \end{tabular} 
\end{listliketab}

\noindent
Note that $\calL(\mho_i) \subseteq \calL(\mho)$, meaning that $\calI$~is also a multiworld interpretation of each\/~$\mho_i$ into~$\calM$.\looseness=-1	
\end{definition}

We define the label-erasing function $\form$ from multiformulas to formulas, as well as multiformula equivalence, and list some of the latter's easily provable properties.

\begin{definition}
The \emph{function $\form$} from multiformulas to formulas is defined as follows:
\ifarxiv
\begin{align*}
\form(\sigma : \varphi) & \ce \varphi,
\\
\form(\mho_1 \spconj \mho_2) & \ce \form(\mho_1) \land \form(\mho_2),
\\ 
\form(\mho_1 \spdisj \mho_2) &\ce \form(\mho_1) \lor \form(\mho_2).
\end{align*}
\else
$
\form(\sigma : \varphi)\!\ce\! \varphi$, \!\!   $\form(\mho_1 \spconj \mho_2)\!\ce \!\form(\mho_1) \land \form(\mho_2)$, \!\! \mbox{$\form(\mho_1 \spdisj \mho_2)\!\ce\! \form(\mho_1) \lor \form(\mho_2)$}.
\fi
\end{definition}

\begin{definition}[Multiformula equivalence]
Multiformulas\/~$\mho_1$~and\/~$\mho_2$ are \emph{equivalent},  denoted\/ $\mho_1 \equiv \mho_2$, if{f} $\calL(\mho_1)=\calL(\mho_2)$ and $\calM, \calI \vDash \mho_1\,\Leftrightarrow\,\calM, \calI \vDash \mho_2$ for any multiworld interpretation~$\calI$ of\/~$\mho_1$ into a model~$\calM$.
\end{definition}

\begin{lemma}[Equivalence property] \label{Lemma: Multiformula properties}
For any multiformula\/~$\mho$, label~$\sigma$, and  formulas~$\varphi$~and~$\psi$,\looseness=-1
\ifarxiv
\begin{itemize}
    \item $\mho \spconj \mho \equiv \mho \spdisj \mho \equiv \mho$,    
    \item
    $\sigma : \varphi \spconj \sigma : \psi \equiv \sigma : (\phi \wedge \psi)$, and
    \item 
    $\sigma : \varphi \spdisj \sigma : \psi \equiv \sigma : (\phi \vee \psi)$.
   \end{itemize}
\else
\\
    $\mho \spconj \mho \equiv \mho \spdisj \mho \equiv \mho$,    and 
    $\sigma : \varphi \spconj \sigma : \psi \equiv \sigma : (\phi \wedge \psi)$, and  $\sigma : \varphi \spdisj \sigma : \psi \equiv \sigma : (\phi \vee \psi)$.
\fi
\end{lemma}

\begin{lemma}[Normal forms]
\label{lem:sdnf}
For each multiformula\/~$\mho$, there exists an equivalent multiformula\/~$\mho^d$~$(\mho^c)$ in~SDNF\/~$($SCNF\/$)$  such that\/ $\mho^d$~$(\mho^c)$~is a\/ $\spdisj$-disjunction\/~$(\spconj$-conjunction\/$)$ of\/ $\spconj$-conjunctions\/~$(\spdisj$-disjunctions\/$)$ of labeled formulas~$\sigma : \varphi$ such that each disjunct\/~$($conjunct\/$)$ contains exactly one occurrence of each label $\sigma \in \calL(\mho)$.
\end{lemma}
\begin{proof}
Since $\spdisj$~and~$\spconj$~behave classically, one can employ the standard transformation into the~DNF/CNF. In order to ensure one label per disjunct/conjunct rule, multiple labels can be combined using Lemma~\ref{Lemma: Multiformula properties}, whereas missing labels can be added in the form of~$\sigma : \bot$~($\sigma : \top$).  \looseness=-1
\end{proof}

We now introduce the uniform interpolation property for nested sequents. Here, the uniform interpolants are multiformulas instead of formulas.

\begin{definition}[NUIP]
\label{def:uip nested sequents}
Let a nested sequent calculus\/~${\sf NL}$ be sound and complete w.r.t.~a logic\/~$\L$. We say that\/  {\sf NL}~has the \emph{nested-sequent uniform interpolation property}, or~\emph{NUIP}, if for each nested sequent\/~$\Gamma$ and atomic proposition~$p$ there exists  a multiformula~$A_p(\Gamma)$, called a \emph{nested uniform interpolant}, such that 
\begin{enumerate}[(i)]
\item\label{NUIP:1} $ \Var\bigl(A_p(\Gamma)\bigr) \subseteq  \Var(\Gamma) \setminus \{ p \}$ and $\calL\bigl(A_p(\Gamma)\bigr) \subseteq \calL(\Gamma)$;
\item\label{NUIP:2} for each  multiworld interpretation~$\calI$ of\/~$\Gamma$ into an\/ $\L$-model~$\calM$  
\[
\calM, \calI \models A_p(\Gamma)\quad\text{implies}\quad \calM, \calI \models \Gamma;
\]
\item\label{NUIP:3} for each nested sequent~$\Sigma$ with $p \notin \Var(\Sigma)$ and $\calL(\Sigma)=\calL(\Gamma)$ and for each  multiworld interpretation~$\calI$ of\/~$\Gamma$ into an\/ $\L$-model~$\calM$, 
\[
\calM, \calI \not \models A_p(\Gamma) \text{ and } \calM, \calI \not \models \Sigma
\quad\text{imply}\quad
\calM', \calI' \not \models \Gamma \text{ and } \calM', \calI' \not \models \Sigma
\]
for some  multiworld interpretation~$\calI'$ of\/~$\Gamma$ into some\/ $\L$-model~$\calM'$.
\end{enumerate}
\end{definition}

The condition on labels in~\eqref{NUIP:1} ensures that interpretations of~$\Gamma$ are suitable for~$A_p(\Gamma)$.\looseness=-1

\begin{remark}
B\'{i}lkov\'{a}'s definition  in~\cite{Bilkova11} differs in several ways. Apart from a minor difference in condition~\eqref{UIP:3}, our definition involves semantic notions and uses multiformula interpolants instead of  formulas. 
\end{remark}

\begin{lemma}\label{lem:uip equiv}
If a nested calculus\/~{\sf NL} has the NUIP, then its logic\/~$\L$ has the UIP.  \looseness=-1
\end{lemma}
\begin{proof}
To show the existence of~$\forall p\phi$, consider a nested uniform interpolant~$A_p(\varphi)$ of the nested sequent~$\phi$, with $\calL(\phi) =\{1\}$. By Lemma~\ref{lem:sdnf}, w.l.o.g.~we can assume that $A_p(\varphi) = 1 : \xi$. Let $\forall p \varphi \ce\xi$.
We establish the UIP properties based on the corresponding NUIP properties. 

By NUIP\eqref{NUIP:1},  
$ \Var(\forall p \varphi) = \Var (1 : \xi) \subseteq   \Var(\varphi) \setminus \{p\}$ which establishes UIP\eqref{UIP:1} (cf.~Def.~\ref{def:UniformInterpolation}). \looseness=-1

For UIP\eqref{UIP:2} we use a semantic argument. Assume towards a contradiction that $ \nvdash_\L \xi \to \varphi$, in which case by completeness $\calM, w \not \models \xi \to \varphi$ for some $\L$-model  $\calM =(W,R,V)$   and  $w\in W$. 
Consider a multiworld interpretation~$\calI$ of sequent $\phi$ into $\calM$ such that $\calI(1)\ce w$. Then $\calM, \calI \models 1 : \xi$ but $\calM, \calI \not\models \phi$, in contradiction to NUIP\eqref{NUIP:2}. Hence, $ \vdash_\L \forall p \varphi \to \varphi$ as required. \looseness=-1

Finally, for UIP\eqref{UIP:3}, let $p \notin\Var(\psi)$ and  suppose $\nvdash_\logic \psi \imp \xi$. Once again, by completeness, $\calM,w \not\models \psi \to \xi$  for some $\L$-model $\calM=(W,R,V)$ and  $w \in W$. 
Consider the nested sequent~$\overline{\psi}$,  with $\calL(\overline{\psi}) = \calL(\phi) =\{1\}$, and a multiworld interpretation~$\calI$ of sequent~$\phi$ into~$\calM$ with $\calI(1)\ce w$. 
Then $\calM, \calI \not \models 1 : \xi$ and $\calM, \calI \not \models \overline{\psi}$. 
By NUIP\eqref{NUIP:3}, there must exist an $\L$-model~$\calM'$ and a multiworld interpretation~$\calI'$ of sequent~$\phi$ into~$\calM'$ such that  $\calM', \calI' \not \models \phi$ and $\calM', \calI' \not\models \overline{\psi}$. In other words, $\calM', \calI'(1) \not \models \varphi$ and $\calM', \calI'(1) \models \psi$. 
Thus, by soundness of~$\L$, we have $\nvdash_\L \psi \to \varphi$, thus completing the proof of UIP\eqref{UIP:3}. \looseness=-1 
\end{proof}

Since we use bisimulations up to~$p$ to find a model~$\calM'$ in the NUIP\eqref{NUIP:3} condition, we replace it with a (possibly) stronger condition~(iii)$'$:
\begin{definition}[BNUIP]\label{def:Nested uip 2} 
A nested sequent calculus\/~$\mathsf{NL}$ has the \emph{bisimulation nested-sequent uniform interpolation property}, or~\emph{BNUIP}, if, in addition to conditions NUIP\/\eqref{NUIP:1}--\eqref{NUIP:2} from  Def.~\ref{def:uip nested sequents},\looseness=-1
\begin{description}
\item[\textup{(iii)$'$}]
\label{BNUIP:3} 
for each\/ $\L$-model~$\calM$ and multiworld interpretation~$\calI$ of\/~$\Gamma$ into~$\calM$, if 
$
\calM, \calI \not \models A_p(\Gamma)$,
then there are an\/ $\L$-model~$\calM'$ and  multiworld interpretation~$\calI'$ of\/~$\Gamma$ into~$\calM'$ such that $(\calM',\calI') \sim_p (\calM, \calI)$ and 
$
\calM', \calI' \not \models \Gamma$.
\end{description}
\end{definition}

It easily follows from Theorem~\ref{thm:invariantform} that, like formulas, both nested sequents and multiformulas are invariant under bisimulations: 

\begin{lemma}\label{lem:invariantsequent}
Let\/ $\Gamma$~$(\mho)$~be a sequent\/~$($multiformula\/$)$ not containing~$p$ and let $\calI$~and~$\calI'$~be multiworld interpretations of~\/$\Gamma$~$(\mho)$ into~$\calM$~and~$\calM'$ respectively such that\/
$(\calM, \calI) \sim_p (\calM', \calI')$. Then  $\calM, \calI \models \Gamma$ if{f} $\calM', \calI' \models \Gamma$ $(\calM, \calI \models \mho$ if{f} $\calM', \calI' \models \mho)$.
\end{lemma}
\ifarxiv
\begin{proof}
If $(\calM, \calI) \sim_p (\calM', \calI')$, then $(\calM, \calI(\sigma)) \sim_p (\calM, \calI'(\sigma))$ for all labels~$\sigma$ in~$\Gamma$~($\mho$). By Theorem~\ref{thm:invariantform} we have $\calM,\calI(\sigma) \models \phi$ if{f} $\calM',\calI'(\sigma) \models \phi$ for all $\sigma : \phi$ in $\Gamma$ ($\mho$). The statements easily follow from Defs.~\ref{def:truth_seq}~and~\ref{def:truth_multif}.
\end{proof}
\fi

\begin{lemma}\label{lem:condition(iii)$'$}
If\/ $\Gamma$~and~$A_p(\Gamma)$~satisfy\/~\textup{(iii)$'$} of Def.~\ref{def:Nested uip 2}, then they satisfy\/~\eqref{NUIP:3}  of Def.~\ref{def:uip nested sequents}.
\end{lemma}
\begin{proof}
Let $\Sigma$~be a nested sequent with $p \notin\Var(\Sigma)$ and $\calL(\Sigma)=\calL(\Gamma)$. 
Let $\calM, \calI \not \models A_p(\Gamma)$ and $\calM, \calI \not \models \Sigma$. 
By BNUIP(iii)$'$ we find an $\L$-model~$\calM'$ and $\calI'$~from~$\Gamma$ into~$\calM'$ such that $(\calM',\calI') \sim_p (\calM, \calI)$ and $\calM', \calI' \not \models \Gamma$. 
By Lemma~\ref{lem:invariantsequent}, we also conclude $\calM', \calI' \not \models \Sigma$. 
\end{proof}

\begin{corollary}\label{lem:bnuip equiv}
 If a nested calculus\/~{\sf NL} has the BNUIP, then its logic\/~$\L$ has the UIP. \looseness=-1
\end{corollary}

\subsection{Uniform interpolation for $\K$}

\begin{table}[t]
  \begin{center}
  \renewcommand*{\arraystretch}{1.4}
    \begin{tabular}{| l l |} \hline
    $\Gamma$~matches \quad \quad \quad 
    & 
    $A_p(\Gamma)$~equals
    \\
    \hline
    $\Gamma'\vlfill{\top}_\sigma$ 
    &
    $\sigma : \top$ 
    \\
    $\Gamma'\vlfill{p,\overline p}_\sigma$ 
    &
    $\sigma : \top$
    \\
    $\Gamma'\vlfill{\varphi \lor \psi}$ 
    &
    $A_p\bigl(\Gamma'\vlfill{\varphi \lor \psi, \varphi, \psi}\bigr)$
    \\
    $\Gamma'\vlfill{\varphi \land \psi}$ 
    &
    $A_p\bigl(\Gamma'\vlfill{\varphi \land \psi, \varphi}\bigr) \spconj A_p\bigl(\Gamma'\vlfill{\varphi \land \psi, \psi}\bigr)$
    \\
    $\Gamma'\vlfill{\Box \varphi}_\sigma$
    &
    $\bigspconj_{i=1}^m \left(\sigma : \Box \delta_i \spdisj \bigspdisj_{\tau \neq \sigma \ast n} \tau : \gamma_{i,\tau} \right)$
    \\
    & where $n$~is the smallest integer such that $\sigma \ast n \notin\calL(\Gamma)$ and
    the SCNF  \\
    & of  $A_p\bigl(\Gamma'\vlfill{\Box \varphi, [\varphi]_{\sigma \ast n}}\bigr)$ is $\bigspconj_{i=1}^m \left(\sigma \ast n : \delta_i  \spdisj \bigspdisj_{\tau \neq \sigma \ast n} \tau : \gamma_{i,\tau} \right)$,
    \\
    $\Gamma'\vlfill{ \Diamond \varphi,[\Delta]_{\sigma*n}}$
    &
    $A_p\bigl(\Gamma'\vlfill{ \Diamond \varphi,[\Delta, \varphi]}\bigr)$  \\ \hline
    \end{tabular}
    \end{center}
\caption{Recursive construction of~$A_p(\Gamma)$  for~{\sf NK} for~$\Gamma$ that are not $\K$-saturated.}
\label{table:Ap}
\end{table}

In this section, we present our method of constructing nested uniform interpolants satisfying BNUIP for the calculus~$\NK$. It is based on Pitts's  method~\cite{Pitts92JSL}. 
Interpolants~$A_p(\Gamma)$ are defined recursively on the basis of the terminating calculus from Fig.~\ref{fig:rules2}. If $\Gamma$~is not $\K$-saturated, $A_p(\Gamma)$~is defined recursively in Table~\ref{table:Ap} based on the form of~$\Gamma$. For rows~2--5, we assume that the formula displayed in the left column is not $\K$-saturated in~$\Gamma$, whereas for~$\Diamond \varphi$ in the last row we assume it not to be $\K$-saturated w.r.t.~$\sigma \ast n$ in~$\Gamma$.\footnote{Strictly speaking, this is a non-deterministic algorithm. Since the order does not affect our results, we do not specify it. However, it is more efficient to apply rows~1--2 of Table~\ref{table:Ap} first and row~5 last.} Each row in the table corresponds to a rule in the proof search\ifarxiv, where the left column in the table corresponds to the conclusion of a rule and the right column uses the premise(s) of the rule\fi.

For $\K$-saturated~$\Gamma$, we define~$A_p(\Gamma)$ recursively as follows: 
\begin{equation}
\label{eq:ap_step_sat}
A_p(\Gamma) \ce 
\bigspdisj_{\substack{ {\sigma : \ell} \in \Gamma \\\ell \in {\sf Lit} \setminus \{p,\overline{p}\} }} \sigma : \ell \qquad\spdisj\qquad 
\bigspdisj_{\substack{ \tau \in \calL(\Gamma)\\ (\exists \psi) \tau : \Diamond \psi \in \Gamma}} \tau : \Diamond A_p^\form\left(\bigvee\nolimits_{\tau : \Diamond \psi \in \Gamma} \psi\right),
\end{equation}
where $A_p^\form(\Gamma) \ce \form\bigl(A_p(\Gamma)\bigr)$. 
Since here we apply $\form$ to the multiformula~$A_p(\Gamma)$ with 1~being its only label,  we have $\calM, \calI \models \mho$ if{f} $\calM, \calI(1) \models \form(\mho)$ for such multiformulas~$\mho$.
\ifarxiv(As usual, we define the empty disjunction to be false, which in this format means $\bigspdisj \emptyset \ce 1 : \bot$.) \fi

The construction of~$A_p(\Gamma)$ is well-defined (modulo a chosen order) because it terminates w.r.t.~the following ordering on nested sequents. For a nested sequent~$\Gamma$, let $d(\Gamma)$~be the number of its distinct diamond subformulas. Let  $\ll$~be the ordering in which the rules of~$\NK$ terminate (see Lemma~\ref{lem:terminationK}). Consider the lexicographical ordering based on the pair~$(d, \ll)$. For each row in Table~\ref{table:Ap}, $d$~stays the same but the recursive calls are for premise(s) lower w.r.t.~ordering~$\ll$. The recursive call in~\eqref{eq:ap_step_sat} for $\K$-saturated sequents, on the other hand, decreases~$d$ because the set of diamond subformulas of $\bigvee_{\tau : \Diamond \psi \in \Gamma} \psi$ is strictly  smaller than that of~$\Gamma$. When $d(\Gamma)=0$ for a $\K$-saturated~$\Gamma$, the second disjunct of~\eqref{eq:ap_step_sat} is empty and, thus, no new recursive calls are generated.

\begin{example}
\label{ex:important}
We use 
Lemmas~\ref{Lemma: Multiformula properties}~and~\ref{lem:sdnf} as necessary.
\begin{enumerate}
\item 
\ifarxiv The algorithm for $A_p(\Box p, \Box \overline{p})$ calls the calculation of $A_p\left(\Box p, \Box \overline{p}, [p]_{11}\right)$, which in turn calls $A_p\left(\Box p, \Box \overline{p}, [p]_{11}, [\overline{p}]_{12}\right)$. The latter sequent is $\K$-saturated, and the algorithm returns $1 : \bot \spdisj 1 : \bot$, the first disjunct corresponding to the empty disjunction of literals other than~$p$~and~$\overline{p}$ and the second one representing the absent diamond formulas. Computing its SCNF we get $A_p\left(\Box p, \Box \overline{p}, [p]_{11}, [\overline{p}]_{12}\right)\equiv1: \bot \spdisj 11 : \bot \spdisj 12 : \bot$. Applying the transformation from the penultimate row of Table~\ref{table:Ap}, we first get \[
A_p\left(\Box p, \Box \overline{p}, [p]_{11}\right)=1: \bot \spdisj 11 : \bot \spdisj 1 : \Box \bot\equiv 1 : \Box \bot \spdisj 11 : \bot,
\] 
and finally $A_p\left(\Box p, \Box \overline{p}\right) = 1 : \Box \bot \spdisj 1 : \Box \bot \equiv 1 : \Box \bot$.
It is easy to check that  $1 : \Box \bot$~is indeed a bisimulation nested uniform interpolant of the nested sequent $\Box p, \Box \overline{p}$ w.r.t.~$p$, and, accordingly, $\Box \bot$~is a uniform interpolant of the formula $\Box p \lor \Box \overline{p}$. 
\else
For the sequent $\Box p, \Box \overline{p}$, the algorithm  saturates $A_p(\Box p, \Box \overline{p})$ to $A_p\left(\Box p, \Box \overline{p}, [p]_{11}\right)$ and, finally, to $A_p\left(\Box p, \Box \overline{p}, [p]_{11}, [\overline{p}]_{12}\right)=1 : \bot \spdisj 1 : \bot$  (the two $1 : \bot$'s are for the absent literals  from ${\sf Lit} \setminus\{p,\overline{p}\}$  and diamond formulas respectively). Computing the SCNF we get $A_p\left(\Box p, \Box \overline{p}, [p]_{11}, [\overline{p}]_{12}\right)=1: \bot \spdisj 11 : \bot \spdisj 12 : \bot$. Applying the transformation from  row~5 of Table~\ref{table:Ap} twice,  $
A_p\left(\Box p, \Box \overline{p}, [p]_{11}\right)=1: \bot \spdisj 11 : \bot \spdisj 1 : \Box \bot\equiv 1 : \Box \bot \spdisj 11 : \bot$
and, finally, $A_p\left(\Box p, \Box \overline{p}\right) = 1 : \Box \bot \spdisj 1 : \Box \bot \equiv 1 : \Box \bot$.
It is easy to check that   $\Box \bot$ is a uniform interpolant of the formula $\Box p \lor \Box \overline{p}$. 
\fi
\item 
Consider the nested sequent $\Gamma=\overline{p}, \Diamond q \land \Diamond p, [q]$. In the absence of boxes, the algorithm amounts to processing the $\K$-saturated sequents in the leaves of the  proof-search tree
\begin{center}
\AXC{$\overline{p}, \Diamond q \land \Diamond p, \Diamond q, [q]_{11}$}
\AXC{$\overline{p}, \Diamond q \land \Diamond p, \Diamond p, [q,p]_{11}$}
\UIC{$\overline{p}, \Diamond q \land \Diamond p, \Diamond p, [q]_{11}$}
\BIC{$\overline{p}, \Diamond q \land \Diamond p, [q]_{11}$}
\DisplayProof
\end{center}
We have 
\begin{align*}
A_p(\overline{p}, \Diamond q \land \Diamond p, \Diamond q, [q]_{11}) &=  11 : q \spdisj 1 : \Diamond A_p^\form(q) \\
A_p(\overline{p}, \Diamond q \land \Diamond p, \Diamond p, [q,p]_{11}) &=   11 : q \spdisj 1 : \Diamond A_p^\form(p).
\end{align*}
Since formulas~$A^\form_p(q)$~and~$A^\form_p(p)$ can be simplified to~$q$~and~$\bot$ respectively, putting everything together yields
$
A_p(\Gamma) \equiv (11 : q \spdisj 1 : \Diamond q  ) \spconj (11 : q \spdisj 1: \Diamond \bot  )$, which is equivalent to~$11:q$ since $\Diamond \bot$~can never be true. Again, it is easy to verify that $11: q$~is a bisimulation nested uniform interpolant of $\overline{p}, \Diamond q \land \Diamond p, [q]_{11}$ w.r.t.~$p$. \ifarxiv For instance, if $q$~is false at~$\calI(11)$, then one can falsify the sequent by making $p$~true at~$\calI(1)$ and false everywhere else in the irreflexive intransitive finite treelike model.\fi
\end{enumerate}
\end{example}

\begin{theorem} \label{thm:uipK}
The nested calculus\/~$\NK$ has the BNUIP. 
\end{theorem}

\begin{proof}
It is easy to see that BNUIP\eqref{NUIP:1} is satisfied.
To prove BNUIP\eqref{NUIP:2}, let $\Gamma$~be a nested sequent and $\calI$~be a multiworld interpretation of~$\Gamma$ into a $\K$-model $\calM=(W,R,V)$ such that $\calM, \calI \models A_p(\Gamma)$ (by BNUIP\eqref{NUIP:1},  $\calI$~is suitable for~$A_p(\Gamma)$). We show $\calM, \calI \models \Gamma$
by induction on the nested sequent ordering~$(d, \ll)$.   Considering the construction of~$A_p(\Gamma)$, we treat the cases of Table~\ref{table:Ap} first  and deal with the case of $\K$-saturated~$\Gamma$ last.
\begin{itemize}
 \ifarxiv 
\item For rows~1--2 of Table~\ref{table:Ap}, both $\Gamma = \Gamma' \{p, \overline{p} \}_\sigma$ and $\Gamma= \Gamma'\{\top\}_\sigma$ hold in all models, under all interpretations.

\item For row~3, if $\Gamma=\Gamma'\{\varphi \lor \psi \}_\sigma$ and $\calM, \calI \models A_p(\Gamma'\{\varphi \lor \psi, \varphi, \psi \}_\sigma)$, by induction hypothesis, we have $\calM, \calI \models \Gamma'\{\varphi \lor \psi, \varphi, \psi \}_\sigma$. Then $\calM, \calI \models \Gamma'\{\varphi \lor \psi\}$ since either of $\calM, \calI(\sigma) \models \varphi$ or $\calM, \calI(\sigma) \models \psi$ implies $\calM, \calI(\sigma) \models \varphi \lor \psi$.

\item For row~4, if $\Gamma=\Gamma'\{\varphi \land \psi \}$ and  $\calM, \calI \models A_p(\Gamma'\{\varphi \land \psi, \varphi \}) \spconj A_p(\Gamma'\{\varphi \land \psi, \psi \})$, by induction hypothesis,  $\calM, \calI \models \Gamma'\{\varphi \land \psi, \varphi \}$ and $\calM, \calI \models \Gamma'\{\varphi \land \psi, \psi \}$. Hence, $\calM, \calI \models \Gamma'\{\varphi \land \psi\}$. 
\else
\item Cases in rows~1--2 of Table~\ref{table:Ap} are trivial. Those in rows~3, 4,~and~6 are similar, so we only discuss row 6.
\fi

\item For row~6, if $\Gamma=\Gamma'\{\Diamond \varphi, [\Delta]_{\sigma\ast n}  \}$ and $\calM, \calI \models A_p(\Gamma'\{\Diamond \varphi, [\Delta, \varphi]_{\sigma\ast n}  \})$, by induction hypothesis, $\calM, \calI \models \Gamma'\{\Diamond \varphi, [\Delta, \varphi]_{\sigma\ast n}  \}$. Since $\calM, \calI(\sigma \ast n) \models \varphi$ implies $\calM, \calI(\sigma) \models \Diamond \varphi$,  it follows that $\calM, \calI \models \Gamma'\{ \Diamond \varphi, [\Delta]_{\sigma\ast n} \}$.

\item For row~5, let $\Gamma=\Gamma' \{ \Box \varphi\}_\sigma$, and $A_p\bigl(\Gamma'\vlfill{\Box \varphi, [\varphi]_{\sigma \ast n}}\bigr)\equiv\bigspconj_{i=1}^m \left(\sigma \ast n : \delta_i  \spdisj \bigspdisj_{\tau \neq \sigma \ast n} \tau : \gamma_{i,\tau} \right)$ for some $\sigma \ast n \notin\calL(\Gamma)$, and  
\begin{equation}
\label{eq:K_box}
\calM, \calI \models \bigspconj_{i=1}^m \left(\sigma : \Box \delta_i \spdisj \bigspdisj_{\tau \neq \sigma \ast n} \tau : \gamma_{i,\tau} \right).
\end{equation}
For any~$v$ such that~$\calI(\sigma) R v$, define a multiworld interpetation $\calI_v \ce \calI \sqcup \{ (\sigma \ast n, v)\}$ of 
$\Gamma'\vlfill{\Box \varphi, [\varphi]_{\sigma \ast n}}$ into~$\calM$. 
It follows from~\eqref{eq:K_box} that, for each~$i$, either $\calM, \calI_v(\tau) \models \gamma_{i,\tau}$ for some $\tau \in \calL(\Gamma)$ or $\calM, \calI_v(\sigma \ast n) \models \delta_i$, 
meaning that $
\calM, \calI_v \models A_p(\Gamma'\vlfill{\Box \varphi, [\varphi]_{\sigma \ast n}})
$. By induction hypothesis, 
$
\calM, \calI_v  \models \Gamma'\vlfill{\Box \varphi, [\varphi]_{\sigma \ast n}}
$ whenever~$\calI(\sigma) R v$. 
Clearly, $\calM, \calI \models \Gamma$ if $\calM, \calI(\sigma) \models \Box \varphi$. 
Otherwise, there exists a~$v$ such that $\calI(\sigma) R v$~and $\calM, v \not\models \varphi$. For this world  $
\calM, \calI_v  \models \Gamma'\vlfill{\Box \varphi, [\varphi]_{\sigma \ast n}}
$  implies $
\calM, \calI_v  \models \Gamma'\vlfill{\Box \varphi}_\sigma
$, which yields $
\calM, \calI  \models \Gamma
$ because  $\calI_v$~agrees with~$\calI$ on all labels from~$\Gamma$. 

\item 
Finally, let $\Gamma$~be $\K$-saturated and $\calM, \calI \models A_p(\Gamma)$ from~\eqref{eq:ap_step_sat}.
Clearly, $\calM, \calI \models \Gamma$ if we have $\calM, \calI(\sigma) \models \ell$ for some~$\sigma : \ell \in \Gamma$. Thus, it only remains to consider the case when $\calM, \calI(\tau) \models \Diamond A_p^\form\left(\bigvee_{\tau : \Diamond \psi \in \Gamma} \psi\right)$
for some  $\tau \in \calL(\Gamma)$. 
Then $\calM, v \models A_p^\form\left(\bigvee_{\tau : \Diamond \psi \in \Gamma} \psi\right)$ for some~$v$ such that~$\calI(\tau)Rv$ and, accordingly, $
\calM, \calJ  \models A_p\left(\bigvee_{\tau : \Diamond \psi \in \Gamma} \psi\right)
$ for  $\calJ\ce \{(1,v)\}$.
By induction hypothesis (for smaller~$d$),  
$
\calM, \calJ \models  \bigvee_{\tau : \Diamond \psi \in \Gamma} \psi$,  and, hence, $\calM, v\models\psi$ for some  $\tau : \Diamond \psi \in \Gamma$. Now $\calM, \calI \models \Gamma$ follows from~$\calI(\tau)Rv$. This case concludes the proof for~\eqref{NUIP:2}.\looseness=-1
\end{itemize}

It only remains to  prove BNUIP(iii)$'$. 
Let  $\calI$~be a multiworld interpretation  of~$\Gamma$ into a $\K$-model~$\calM$ such that $\calM, \calI \not \models A_p(\Gamma)$. We must find another multiworld interpretation~$\calI'$ into some $\K$-model~$\calM'$   such that $(\calM',\calI') \sim_p (\calM,\calI)$ and $\calM', \calI' \not \models \Gamma$. We construct these~$\calM'$~and~$\calI'$ while simultaneously  proving BNUIP(iii)$'$ by induction on the lexicographic order~$(d, \ll)$. Recall that $\K$-models (and their submodels)  are irreflexive intransitive trees.
\begin{itemize}
\item 
Let $\Gamma$~be $\K$-saturated and $\calM, \calI \not\models A_p(\Gamma)$ for~$A_p(\Gamma)$ from~\eqref{eq:ap_step_sat}.
We first briefly sketch the construction and the proof. 
The labeled literals~$\sigma : \ell$ from~\eqref{eq:ap_step_sat} are used to determine the requisite truth values of atomic propositions other than~$p$ in the worlds from $\mathrm{Range}(\calI)$. 
With that in place, saturation conditions typically take care of the appropriate truth values for compound formulas, with the exception of diamond formulas. 
By contrast, truth values of~$p$ are not (and cannot be) specified in~$A_p(\Gamma)$. To refute~$\Gamma$, they  must generally be adjusted on a world-by-world basis, which prompts the additional requirement that $\calI'$~be injective\footnote{It must be injective as a function, i.e.,~$\calI'(\sigma) = \calI'(\tau)$ implies $\sigma = \tau$.} in order to avoid incompatible requirements on the truth value of~$p$ in a world $\calI(\sigma) =\calI(\tau)$ that originates from distinct nodes~$\sigma$~and~$\tau$. 
Finally, for~$\Diamond \varphi$ to be false at a world $w \in \mathrm{Range}(\calI)$, one must  falsify~$\varphi$ in all children of~$w$, including those  outside $\mathrm{Range}(\calI)$. 
This is achieved by replacing subtrees rooted in these ``out-of-range'' children with bisimilar models obtained by the induction hypothesis from the right disjunct of~\eqref{eq:ap_step_sat},
as  schematically depicted in Fig.~\ref{fig:construction}. We now describe it in detail and prove that it falsifies~$\Gamma$.\looseness=-1
\begin{enumerate}[(1)]
\item
\label{step:injectify}
First, we make the interpretation injective.
It is easy to see (though  tedious to describe in detail) that by a breadth-first recursion on nodes~$\sigma$ in~$\Gamma$, one can duplicate~$\calM_{\calI(\sigma\ast n)}$ according to Def.~\ref{def:copymodels} whenever $\calI(\sigma*n) = \calI(\sigma\ast m)$ for some $m <n$ to  obtain a model~$\calN$ and an injective multiworld interpretation~$\calJ$ of~$\Gamma$ into it such that $(\calN,\calJ) \sim_p (\calM,\calI)$. Thus, $\calJ(\sigma) \ne \calJ(\tau)$ whenever $\sigma \ne \tau$ and $\calN, \calJ \not \models A_p(\Gamma)$ by Lemma~\ref{lem:invariantsequent}.

\begin{figure}[t]
\begin{center}
\ \ \ \ \ \ \ \ \ \ \ \ \ 
\begin{tikzpicture}[modal]
\node[point] (z) [label=below:{\footnotesize{$\calI(\sigma)$}}] { };
\node[world] (z1) [above =of z, yshift=-9.5mm ] {};
\node[point] (v) [above left =of z, xshift=-2mm,label=below:{\footnotesize{$v$}}] {};
\node[itria] (Mv) [above =of v, yshift=6mm]{\footnotesize{$\calM_v$} };
\node[point] (y) [above right =of z, xshift=2mm, label={[right, yshift=1.5mm, xshift=1mm]\footnotesize{$\calI(\sigma*n)$}},label={[right, yshift=-2mm, xshift=1mm]\footnotesize{$\calI(\sigma*m)$}}]{ };
\node[world] (y1) [above =of y, yshift=-9.5mm ] {};
\node[itria] (Yw) [above =of y, yshift=6mm]{\footnotesize{$\calM_w$} };

\path[->] (z) edge  (v);
\path[->] (z) edge  (y);

\node[ ] (model) [below =of z] {in $\calM $};
\end{tikzpicture} \ \ \ 
\begin{tikzpicture}[modal]
\node[point] (z) [label=below:{\footnotesize{$\calI'(\sigma)$}}] { };
\node[world] (z1) [above =of z, yshift=-9.5mm ] {};
\node[point] (v) [above left =of z, xshift=-12mm,label=below:{\footnotesize{$\rho_{\sigma,v}$}}] {};
\node[itria] (Mv) [above =of v, yshift=6mm]{\footnotesize{$\calN_{\sigma,v}$} };
\node[point] (w) [above =of z, yshift=-2.5mm , label={[right, xshift=1mm]\footnotesize{$\calI'(\sigma*n)$}}]{ };
\node[world] (w1) [above =of w, yshift=-9.5mm ] {};
\node[itria] (Mw) [above =of w, yshift=6mm]{\footnotesize{$\calM_w$} };
\node[point] (y) [above right =of z, xshift=12mm,label={[right, xshift=1mm]\footnotesize{$\calI'(\sigma*m)$}}]{ };
\node[world] (y1) [above =of y, yshift=-9.5mm ] {};
\node[itria] (Yw) [above =of y, yshift=6mm]{\footnotesize{$\calM^c_w$} };

\path[->] (z) edge  (v);
\path[->] (z) edge  (w);
\path[->] (z) edge  (y);

\node[ ] (model) [below =of z] {in $\calM'$};
\node[ ] (construction) [left =of v] {\ $\leadsto$ \ };
\end{tikzpicture}

\caption{Main transformations for constructing  model $\calM'$: circles represent worlds in ${\sf Range}(\calI)$.}
\label{fig:construction}
\end{center}
\end{figure}

\item
\label{step:stray_nodes}
Then we deal with out-of-range children.
A model~$\calN'$ is constructed from~$\calN$  by applying the following $\Diamond$-processing step for each node $\tau \in \calL(\Gamma)$ that contains at least one formula of the form~$\Diamond \phi$ (nodes can be chosen in any order). Start by setting $\calN^0 \ce \calN$ and $j\ce 0$: 
\begin{itemize}
 \item{\textbf{$\Diamond$-processing step for~$\tau$:}}
 Since $\calN^j, \calJ \not \models A_p(\Gamma)$, it follows from the second disjunct in~\eqref{eq:ap_step_sat} that $\calN^j, \calJ(\tau) \not\models \Diamond A_p^\form\left(\bigvee_{\tau : \Diamond \psi \in \Gamma} \psi\right)$. 
 Thus, $\calN^j, v \not\models A_p^\form\left(\bigvee_{\tau : \Diamond \psi \in \Gamma} \psi\right)$ for any child~$v$ of~$\calJ(\tau)$ in~$\calN^j$, and, accordingly,  $\calN^j_v, \calI_v \not \models A_p\left(\bigvee_{\tau : \Diamond \psi \in \Gamma} \psi\right) $
 for the multiworld interpretation $\calI_v\ce \{(1,v)\}$ of sequent $\bigvee_{\tau : \Diamond \psi \in \Gamma} \psi$ into the  subtree~$\calN^j_v$ of~$\calN^j$ with root~$v$. 
By the induction hypothesis for a smaller~$d$, there exists a $\K$-model~$\calN_{\tau, v}$ with root~$\rho_{\tau,v}$ such that $(\calN^j_v, v) \sim_p (\calN_{\tau,v}, \rho_{\tau,v})$ and $\calN_{\tau,v}, \rho_{\tau,v} \not\models \bigvee_{\tau : \Diamond \psi \in \Gamma} \psi$. 
 Let $\calN^{j+1}$~be the result of replacing each subtree~$\calN^j_v$ for  children~$v$ of~$\calJ(\tau)$ not in ${\sf Range}(\calJ)$ with~$\calN_{\tau,v}$ in~$\calN^j$ according to Def.~\ref{def:copymodels}. Note that all these subtrees are disjoint because the models are intransitive trees and, hence, these replacements do not interfere with one another. Note also that since ${\sf Range}(\calJ)$ is downward closed and the roots of the replaced subtrees are outside, no world from the range is modified. Thus, $\calJ$~remains an injective interpretation into~$\calN^{j+1}$. Finally, it follows from  Lemma~\ref{lem: Bisimilarity of Submodels} that $(\calN^j,\calJ) \sim_p (\calN^{j+1},\calJ)$. Hence, $\calN^{j+1},\calJ \not \models A_p(\Gamma)$.
\end{itemize}
Let $\calN' = (W',R',V')$ be  the model obtained after  replacements for all~$\tau$'s are completed (again they do not interfere with each other). 
Then $(\calN, \calJ) \sim_p (\calN', \calJ)$
and,  for each out-of-range child~$v$ of~$\calJ(\tau)$ in~$\calN$, the world~$\rho_{\tau,v}$ is  a child of~$\calJ(\tau)$ in~$\calN'$ and $\calN', \rho_{\tau,v} \not \models\bigvee_{\tau : \Diamond \psi \in \Gamma} \psi$. This accounts for all children of~$\calJ(\tau)$ in~$\calN'$.

\item 
\label{step:wiggle_p}
It remains to adjust the truth values of~$p$.
We define $\calM'\ce (W',R', V'_p)$  by modifying the valuation~$V'$ of~$\calN'$ as follows:
\[
V_p'(q) \ce 
\begin{cases}
V'(q) & \text{if $q \ne p$;}
\\
V'(p) \cap \bigl(W' \setminus {\sf Range}(\calJ)\bigr) \sqcup \{v \in W' \mid \exists \sigma (v = \calJ(\sigma) \& \sigma \colon \overline{p} \in \Gamma)\}
& \text{if $q = p$.}
\end{cases}
\]
For $\calI' \ce \calJ$, it immediately follows from the definition that 
\begin{align}
\label{case:iiibpbar}
\calM', \calI'(\sigma) \not\models  \overline{p} &\text{ whenever } \sigma : \overline{p} \in \Gamma;
\\
\label{case:iiibp}
\calM', \calI'(\sigma) \not\models p &\text{ whenever } \sigma : p \in \Gamma
\end{align}
(the latter follows from the injectivity of~$\calI'$ and $\Gamma$~being $\K$-saturated). Moreover, since subtrees~$\calM'_{\rho_{\tau,v}}$ are disjoint from ${\sf Range}(\calI')$, 
\begin{equation}
\label{case:iiibstray}
\calM', \rho_{\tau,v} \not \models \psi \text{ whenever $\tau : \Diamond \psi \in \Gamma$}. 
\end{equation}
\end{enumerate}

After these three steps, we have a model  $(\calM',\calI') \sim_p (\calN', \calJ) \sim_p (\calN,\calJ) \sim_p (\calM,\calI)$ that satisfies~\eqref{case:iiibpbar}, \eqref{case:iiibp},~and~\eqref{case:iiibstray}.
It remains to prove that  $\calM',\calI' \not \models \Gamma$ by showing that $\calM',\calI'(\sigma) \not \models  \varphi$ for all $\sigma : \varphi \in \Gamma$, which is done by induction on the structure of~$\varphi$. 
For $\varphi=\bot$ this is trivial, while $\top$~cannot occur in a $\K$-saturated sequent. 
For $\varphi \in \{p,\overline{p}\}$, this follows from~\eqref{case:iiibpbar}~and~\eqref{case:iiibp}. For any other literal $\phi\in{\sf Lit} \setminus\{p,\overline{p}\}$, according  to~\eqref{eq:ap_step_sat}, $\calM, \calI(\sigma) \not\models \varphi$ because $\calM, \calI \not\models A_p(\Gamma)$, 
which transfers to~$\calM'$~and~$\calI'$ by bisimilarity up to~$p$. For compound formulas other that diamonds, the statement follows by the saturation of~$\Gamma$. For instance, if $\sigma : \Box \psi \in \Gamma$, we get $\sigma \ast n : \psi\in\Gamma$ for some label~$\sigma \ast n$ by $\K$-saturation. 
By induction hypothesis, $\calM', \calI'(\sigma \ast n) \not \models  \psi$. Since $\calI'(\sigma) R' \calI'(\sigma \ast n)$, we conclude  $\calM', \calI'( \sigma) \not \models  \Box \psi$ as required. 
Finally, let $\sigma: \Diamond \psi \in \Gamma$. To falsify~$\Diamond \psi$ at~$\calI'(\sigma)$, we need to  show that $\calM', u \not \models \psi$ whenever~$\calI'(\sigma) R' u$. If $u = \calI'(\sigma \ast n)$ for some  label $\sigma \ast n \in \calL(\Gamma)$, saturation ensures that $\sigma \ast n : \psi \in \Gamma$, hence, $\calM', u \not \models \psi$ by  induction hypothesis. 
The only other children of~$\calI'(\sigma)$ are $u=\rho_{\sigma, v}$, for which $\calM', u \not \models \psi$ follows from~\eqref{case:iiibstray}. This completes the proof of BNUIP(iii)$'$ for $\K$-saturated sequents.

\item  Now we treat all sequents that are not $\K$-saturated based on Table~\ref{table:Ap}.
$A_p(\Gamma'\{ \top\}_\sigma)=A_p(\Gamma'\{ p, \overline{p}\}_\sigma)=\sigma : \top$, which cannot be false, thus, BNUIP(iii)$'$ for them is vacuously true. 

\item For non-saturated $\Gamma'\{ \varphi \lor \psi\}$, $\Gamma'\{ \varphi \land \psi\}$, and $\Gamma'\{ \Diamond \varphi, [\Delta]\}$, the requisite statement easily follows by induction hypothesis. For instance, for the last of the three, one obtains  $(\calM',\calI') \sim_p (\calM, \calI)$ such that $\calM', \calI' \not \models \Gamma'\{ \Diamond \varphi, [\Delta, \varphi]\}$. Since $\Gamma'\{ \Diamond \varphi, [\Delta]\}$ consists of some of these formulas in the same nodes, clearly it is also falsified by $\calM',\calI'$.

\item For the remaining case, assume $\calM, \calI \not \models A_p(\Gamma' \{\Box \varphi \}_\sigma)$, i.e., 
\begin{equation}
\label{eq:axone}
\calM, \calI \not \models \bigspconj_{i=1}^m \left(\sigma : \Box \delta_i \spdisj \bigspdisj_{\tau \neq \sigma \ast n} \tau : \gamma_{i,\tau} \right)
\end{equation}
where 
\begin{equation}
\label{eq:axtwo}
A_p\bigl(\Gamma'\vlfill{\Box \varphi, [\varphi]_{\sigma \ast n}}\bigr)\equiv\bigspconj_{i=1}^m \left(\sigma \ast n : \delta_i  \spdisj \bigspdisj_{\tau \neq \sigma \ast n} \tau : \gamma_{i,\tau} \right).
\end{equation}
By~\eqref{eq:axone}, for some~$i$, we have
$\calM, \calI(\sigma) \not \models \Box \delta_i$ and $\calM, \calI(\tau) \not \models \gamma_{i,\tau}$ for all $\tau \ne \sigma \ast n$. The former means that 
$\calM, v \not \models \delta_i$ for some~$v$ such that~$\calI(\sigma)Rv$. Therefore,  a multiworld interpretation $\calJ \ce\calI \sqcup \{(\sigma \ast n, v) \}$ of $\Gamma'\vlfill{\Box \varphi, [\varphi]_{\sigma \ast n}}$ into~$\calM$ falsifies~\eqref{eq:axtwo}, and, by  induction hypothesis,  there is a multiworld interpretation~$\calJ'$ into a $\K$-model~$\calM'$   such that $(\calM', \calJ')\sim_p(\calM, \calJ)$ and $\calM', \calJ' \not \models \Gamma'\vlfill{\Box \varphi, [\varphi]_{\sigma \ast n}}$. 
For $\calI' \ce \calJ' \upharpoonright {\sf Dom}(\calI)$, it is easy to see that $(\calM,\calI) \sim_p (\calM', \calI')$ and $\calM', \calI' \not \models \Gamma'\{\Box \varphi \}_\sigma$ because all formulas from $\Gamma'\{\Box \varphi \}_\sigma$ are present in $\Gamma'\vlfill{\Box \varphi, [\varphi]_{\sigma \ast n}}$.
\end{itemize}
This concludes the proof of BNUIP(iii)$'$, as well as of BNUIP.
\end{proof}

This implies the UIP for~$\K$,  first proved by Ghilardi~\cite{Ghilardi95}.\looseness=-1

\begin{corollary}
Logic\/~$\K$ has the uniform interpolation property.
\end{corollary}

\ifarxiv
\begin{remark}
Note that the structure of models as irreflexive intransitive trees was substantially used to ensure that the replacements applied to the original model do not interfere with each other. The fact that each world has at most one parent provided the modularity necessary to implement various requirements on the sequent-refuting model.
\end{remark}
\fi

\begin{example}
 In Example~\ref{ex:important} we saw that $A_p(\Box p, \Box \overline{p}) \equiv 1:\Box \bot$. We now use this example to demonstrate the importance of injectivity in BNUIP(iii)$'$. Indeed, suppose $\calM, \calI \not \models 1 : \Box \bot$, i.e., $\calI(1)$~has at least one child. Assume this is the only child, as in a  model  depicted on the left:\looseness=-1
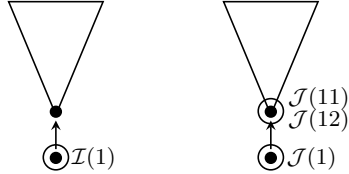
\begin{center}
\begin{tikzpicture}[modal]
\node[point] (z) [label=right:{\footnotesize{$\calI(1)$}}] {};
\node[world] (z1) [above =of z, yshift=-9.5mm ] {};
\node[point] (w) [above =of z, yshift=-2.5mm]{ };
\node[itria] (Mw) [above =of w, yshift=6mm]{ } ;

\path[->] (z) edge  (w);

\end{tikzpicture}
\ \ \ \ \ \ \ \ \ \ 
\begin{tikzpicture}[modal]
\node[point] (z) [label=right:{\footnotesize{$\calJ(1)$}}] {};
\node[world] (z1) [above =of z, yshift=-9.5mm ] {};
\node[point] (w) [above =of z, yshift=-2.5mm, label={[right, yshift=1mm, xshift=1mm]\footnotesize{$\calJ(11)$}},label={[right, yshift=-2mm, xshift=1mm]\footnotesize{$\calJ(12)$}}]{ };
\node[world] (w1) [above =of w, yshift=-9.5mm ] {};
\node[itria] (Mw) [above =of w, yshift=6mm]{ } ;

\path[->] (z) edge  (w);

\end{tikzpicture}
\end{center}
For a saturation $\Box p, \Box \overline{p}, [p]_{11}, [\overline{p}]_{12}$ of this sequent, we found  an interpolant in SCNF: namely, $1: \bot \spdisj 11 : \bot \spdisj 12 : \bot$. A multiworld interpretation~$\calJ$ mapping both~$11$~and~$12$ to the only child of $\calJ(1) \ce \calI(1)$ yields the picture on the right. Clearly, the SCNF is false:  $\calM, \calJ \not \models 1: \bot \spdisj 11 : \bot \spdisj 12 : \bot$. But, without forcing~$\calJ$ to be injective, it is impossible to make $\Box p, \Box \overline{p}$ false at~$\calJ(1)$: whichever truth value $p$~has at~$\calJ(11)$, it makes one of the boxes true.
\end{example}

\subsection{Uniform interpolation for~$\D$~and~$\T$}

\begin{table}[t]
  \begin{center}
  \renewcommand*{\arraystretch}{1.4}
    \begin{tabular}{| l l |} \hline
    $\Gamma$~matches \quad \quad \quad 
    & 
    $A_p(\Gamma)$~equals
    \\
    \hline
    $\Gamma'\vlfill{\Diamond \varphi}$ in logic~$\T$
    &
    $A_p(\Gamma'\vlfill{\Diamond \varphi, \varphi})$ 
   \\
     $\Gamma'\vlfill{\Diamond \varphi}_{\sigma}$ in logic~$\D$
    &
    $\bigspdisj_{i=1}^m \left(\sigma : \Diamond \delta_i \spconj \bigspconj_{\tau \neq \sigma \ast 1} \tau : \gamma_{i,\tau}\right)$  where  the SDNF of 
   \\
   & $A_p(\Gamma'\vlfill{\Diamond \varphi, [\varphi]_{\sigma \ast 1}})$ is $\bigspdisj_{i=1}^m \left( \sigma \ast 1 : \delta_i \spconj \bigspconj_{\tau \neq \sigma \ast 1}\tau : \gamma_{i,\tau} \right)$
    \\ \hline
    \end{tabular}
    \end{center}
\caption{Additional recursive rules for constructing~$A_p(\Gamma)$  for~$\Gamma$ that are not $\T$-saturated (top row) or not $\D$-saturated (bottom row).}
\label{table:ApDT} 
\end{table}

The proof for~$\K$ can be adjusted to prove the same result for~$\D$~and~$\T$.

\begin{theorem} \label{thm:uipDT}
The nested sequent calculi\/~$\ND$~and\/~$\NT$ have the BNUIP.
\end{theorem}

\begin{proof} 
We follow the structure of the proof of Theorem~\ref{thm:uipK} for~$\K$ and  only describe deviations from it. 
If $\Gamma$~is not $\D$-saturated~($\T$-saturated), then cases in Table~\ref{table:Ap} are appended with the bottom row~(top~row) of Table~\ref{table:ApDT}, which is applied only if $\Diamond \varphi$~is not $\D$-saturated~($\T$-saturated) in~$\Gamma$. For $\D$-/$\T$-saturated~$\Gamma$, we define~$A_p (\Gamma)$ by~\eqref{eq:ap_step_sat} as in the previous section.  BNUIP\eqref{NUIP:1} is clearly satisfied by either row in Table~\ref{table:ApDT}.

Let us first show BNUIP\eqref{NUIP:2} for~$\NT$. Although $\T$-models are reflexive, this does not affect the reasoning for either saturated sequents or  non-saturated box formulas. The only new case is applying the top row of Table~\ref{table:ApDT} to a non-$\T$-saturated $\sigma : \Diamond \varphi$ in~$\Gamma$.  Assume $\calM,\calI \models A_p(\Gamma'\vlfill{\Diamond \varphi, \varphi}_\sigma)$ for a $\T$-model~$\calM$.  By induction hypothesis, $\calM, \calI \models \Gamma'\vlfill{\Diamond \varphi, \varphi}_\sigma$. Since $\calM, \calI(\sigma) \models \varphi$ implies $\calM, \calI(\sigma) \models \Diamond \varphi$ by reflexivity, the desired
 $\calM, \calI \models\Gamma'\vlfill{\Diamond \varphi}_\sigma$ follows.\looseness=-1

  For  BNUIP(iii)$'$ for $\T$-saturated sequents, we have to modify the construction in step~\eqref{step:injectify} on p.~\pageref{step:injectify} of an injective multiworld interpretation~$\calJ$ into a new $\T$-model~$\calN$ out of the given~$\calI$ into~$\calM$ where $\calM, \calI \not\models A_p(\Gamma)$. 
In the case of~$\K$, the breadth-first order of injectifying the interpretations of sequent nodes could only yield one situation of~$\sigma \ast n$ being conflated with some already processed~$\tau$: namely, when $\tau = \sigma \ast m$ is a sibling. 
This  can still happen  for $\T$-models and is processed the same way. But, due to  reflexivity,  there is now another possibility: conflating with the parent $\tau = \sigma$. 
In this case, cloning is used\ifarxiv{} (see Fig.~\ref{fig:constructionT})\fi{}  instead~of or in addition~to duplication, which produces a bisimilar $\T$-model by Lemma~\ref{lem: Bisimilarity of Submodels}. 
Having intransitive trees that are reflexive rather than irreflexive  in  step~\eqref{step:stray_nodes} on p.~\pageref{step:stray_nodes}  does not affect the argument.  
The proof that  $\calM', \calI' \not\models \Gamma$ for the given $\T$-saturated~$\Gamma$ in step~\eqref{step:wiggle_p} on p.~\pageref{step:wiggle_p} requires an adjustment only  for the case of $\sigma : \Diamond \psi \in \Gamma$. It is additionally necessary to show that $\calM', \calI'(\sigma) \not\models \psi$ for the reflexive loop at~$\calI'(\sigma)$.
This  is resolved by observing that $\sigma : \psi \in \Gamma$ due to $\T$-saturation and, hence, $\psi$~must also be false in~$\calI'(\sigma)$ by induction hypothesis. 
\ifarxiv
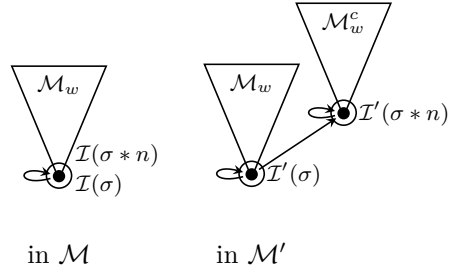
\begin{figure}[t]
\begin{center}
\ \ \ \ \ \ \ \ \ \ \ \ \ 
\begin{tikzpicture}[modal]
\node[point] (w) [label={[right, yshift=2mm, xshift=1mm]\footnotesize{$\calI(\sigma*n)$}},label={[right, yshift=-2mm, xshift=1mm]\footnotesize{$\calI(\sigma)$}}]{ };
\node[world] (w1) [above =of w, yshift=-9.5mm ] {};
\node[itria] (Mw) [above =of w, yshift=6mm]{\footnotesize{$\calM_w$}} ;

\path [->] (w) edge[loop left] (w);

\node[ ] (model) [below =of w] {in $\calM $};
\end{tikzpicture} \ \ \ 
\begin{tikzpicture}[modal]
\node[point] (w) [label=right:{\footnotesize{$\calI'(\sigma)$}}]{ };
\node[world] (w1) [above =of w, yshift=-9.5mm ] {};
\node[itria] (Mw) [above =of w, yshift=6mm]{\footnotesize{$\calM_w$}} ;
\node[point] (y) [above right =of w, xshift=6mm, yshift=2mm, label=right:{\footnotesize{$\calI'(\sigma*n)$}}] { };
\node[world] (y1) [above =of y, yshift=-9.5mm ] {};
\node[itria] (My) [above =of y, yshift=6mm]{\footnotesize{$\calM^c_w$}} ;

\path[->] (w) edge  (y);
\path [->] (w) edge[loop left] (w);
\path [->] (y) edge[loop left] (y);

\node[ ] (model) [below =of w] {in $\calM' $};
\end{tikzpicture}

\caption{Additional transformation for constructing $\T$-model~$\calM'$ for reflexive nodes: cloning.}
\label{fig:constructionT}
\end{center}
\end{figure}
\fi

Finally, for BNUIP(iii)$'$ for non-$\T$-saturated sequents, we gain a new case  when the top row of Table~\ref{table:ApDT} is used, but it is clear that $\calM', \calI' \not \models\Gamma'\vlfill{\Diamond \varphi, \varphi}$ obtained by induction hypothesis directly implies $\calM', \calI' \not \models \Gamma'\vlfill{\Diamond \varphi}$. This completes the proof of BNUIP for~$\NT$.

For BNUIP\eqref{NUIP:2} for~$\ND$, the only new case is applying the bottom row of Table~\ref{table:ApDT} to a  non-$\D$-saturated $\sigma : \Diamond \varphi$ in $\Gamma = \Gamma'\vlfill{\Diamond \varphi}_{\sigma}$. Let 
\ifarxiv\[\else$\fi
\calM, \calI \models \bigspdisj_{i=1}^m \left(\sigma : \Diamond \delta_i \spconj \bigspconj_{\tau \neq \sigma \ast 1} \tau : \gamma_{i,\tau}\right)
\ifarxiv\]\else$\fi{}
for some multiworld interpretation~$\calI$ into a $\D$-model $\calM=(W,R,V)$ where 
\ifarxiv\[\else$\fi
A_p(\Gamma'\vlfill{\Diamond \varphi, [\varphi]_{\sigma \ast 1}})\equiv \bigspdisj_{i=1}^m \left( \sigma \ast 1 : \delta_i \spconj \bigspconj_{\tau \neq \sigma \ast 1}\tau : \gamma_{i,\tau} \right).
\ifarxiv\]\else$\fi{} 
Then, for some~$i$, we have  $\calM, \calI(\tau) \models \gamma_{i, \tau}$ for all $\tau \in \calL(\Gamma)$ and $\calM,\calI(\sigma) \models \Diamond \delta_i$. Thus, $\calM, v \models \delta_i$ for some~$v$ such that~$\calI(\sigma) R v$. 
 Since $\Diamond \varphi$~is not $\D$-saturated in $\Gamma'\{\Diamond \varphi \}_\sigma$, it follows that $\calI_v \ce \calI \sqcup \{ (\sigma \ast 1, v)\}$ is a multiworld interpretation of $\Gamma'\vlfill{\Diamond \varphi, [\varphi]_{\sigma \ast 1}}$ into~$\calM$ such that $\calM, \calI_v \models A_p(\Gamma'\vlfill{\Diamond \varphi, [\varphi]_{\sigma \ast 1}})$. By  induction hypothesis, $\calM, \calI_v \models \Gamma'\vlfill{\Diamond \varphi, [\varphi]_{\sigma \ast 1}}$, from which it easily follows that $\calM, \calI \models \Gamma'\vlfill{\Diamond \varphi}_\sigma$.
 
 \ifarxiv
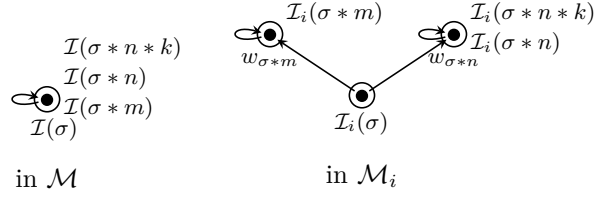
\begin{figure}[t]
\begin{center}
\ \ \ \ \ \ \ \ \ \ \ \ \ 
\begin{tikzpicture}[modal]
\node[point] (w) [label={[right, yshift=2mm, xshift=1mm]\footnotesize{$\calI(\sigma*n)$}},label={[right, yshift=6mm, xshift=1mm]\footnotesize{$\calI(\sigma*n*k)$}},label={[right, yshift=-2mm, xshift=1mm]\footnotesize{$\calI(\sigma*m)$}},label={[below, yshift=-2mm, xshift=1mm]\footnotesize{$\calI(\sigma)$}}]{ };
\node[world] (w1) [above =of w, yshift=-9.5mm ] {};

\path [->] (w) edge[loop left] (w);

\node[ ] (model) [below =of w] {in $\calM $};
\end{tikzpicture} \ \ \ 
\begin{tikzpicture}[modal]
\node[point] (w) [label=below:{\footnotesize{$\calI_i(\sigma)$}}]{ };
\node[world] (w1) [above =of w, yshift=-9.5mm ] {};
\node[point] (y) [above right =of w, xshift=6mm, yshift=2mm, label={[right,yshift=-2mm,xshift=1mm]\footnotesize{$\calI_i(\sigma*n)$}},label={[right,yshift=2mm,xshift=1mm]\footnotesize{$\calI_i(\sigma*n*k)$}},label={[below,yshift=-2mm]\footnotesize{$w_{\sigma*n}$}}] { };
\node[world] (y1) [above =of y, yshift=-9.5mm ] {};
\node[point] (z) [above left =of w, xshift=-6mm, yshift=2mm, label={[right,yshift=2mm,xshift=1mm]\footnotesize{$\calI_i(\sigma*m)$}},label={[below,yshift=-2mm]\footnotesize{$w_{\sigma*m}$}}] { };
\node[world] (z1) [above =of z, yshift=-9.5mm ] {};

\path[->] (w) edge  (y);
\path [->] (y) edge[loop left] (y);
\path[->] (w) edge  (z);
\path [->] (z) edge[loop left] (z);

\node[ ] (model) [below =of w] {in $\calM_i $};
\end{tikzpicture}

\caption{Additional transformation for constructing $\D$-model $\calM_i$ for reflexive leaves.}
\label{fig:constructionT}
\end{center}
\end{figure}
\fi

 For BNUIP(iii)$'$ for $\D$-saturated sequent, we must change step~\eqref{step:injectify}   to preserve $\D$-models.  By Lemma~\ref{lem: Bisimilarity of Submodels}, duplication used for~$\K$ preserves $\D$-models when applied to  non-leaves of  $\D$-models  because they are irreflexive. Now consider the case when $w=\calI(\sigma)$ is a leaf of a model $\calM=(W,R,V)$, but node~$\sigma$ has children in the sequent tree, which $\calI$~can only map to~$w$. To ensure injectivity, we construct an intermediate model~$\calM_i$ separating~$\sigma$ from its children as follows\ifarxiv{} (see Fig.~\ref{fig:constructionT})\fi: 
 \ifarxiv
 \begin{align*}
 W_i &\ce W \sqcup \{w_{\sigma \ast n } \mid  \sigma \ast n \in \calL(\Gamma)\}
 \\
 R_i &\ce R \setminus \{(w,w)\} \sqcup \{(w, w_{\sigma \ast n}), (  w_{\sigma \ast n}, w_{\sigma \ast n})  \mid  \sigma \ast n \in \calL(\Gamma)\}
 \\
 V_i(q) &\ce 
 \begin{cases}
 V(q) \sqcup \{ w_{\sigma \ast n} \mid \sigma \ast n \in \calL(\Gamma)\} & \text{if $w \in V(q)$},
 \\
  V(q) & \text{if $w \notin V(q)$}.
 \end{cases}
 \end{align*}
 \else
 $W_i \ce W \sqcup \{w_{\sigma \ast n } \mid  \sigma \ast n \in \calL(\Gamma)\}$, $R_i \ce R \setminus \{(w,w)\} \sqcup \{(w, w_{\sigma \ast n}), (  w_{\sigma \ast n}, w_{\sigma \ast n})  \mid  \sigma \ast n \in \calL(\Gamma)\}$, and $V_i(q) \ce V(q) \sqcup \{ w_{\sigma \ast n} \mid \sigma \ast n \in \calL(\Gamma)\}$ if $w \in V(q)$ or $V_i(q) \ce V(q)$ if $w \notin V(q)$. 
 \fi 
 Accordingly, $\calI_i(\tau) \ce w_{\sigma \ast n}$ if $\tau$~is a descendant 
of this~$\sigma \ast n$ (or $\sigma \ast n$~itself) or $\calI_i(\tau) \ce \calI(\tau)$ if $\tau$~is not a descendant of any of~$\sigma \ast n$. By reasoning similar to Lemma~\ref{lem: Bisimilarity of Submodels}, it is easy to show that $\calM_i$~is a $\D$-model and $(\calM_i, \calI_i) \sim_p (\calM, \calI)$ with all~$w_{\sigma \ast n}$ being bisimilar to~$w$.  The replacements of step~\eqref{step:stray_nodes}  preserve  $\D$-models by Lemma~\ref{lem: Bisimilarity of Submodels}. Step~\eqref{step:wiggle_p} requires no changes either.
The only subtlety in the proof that  $\calM', \calI' \not\models \Gamma$ for a $\D$-saturated~$\Gamma$ is for  $\sigma : \Diamond \psi \in \Gamma$.
The argument for $\calM', \calI'(\sigma) \not\models \Diamond \psi$ does work  the same way as in~$\K$ for the following reason. Since this~$\Diamond \psi$ is $\D$-saturated,  node~$\sigma$ must have a child in the sequent tree. Injectivity of the constructed~$\calI'$ means that $\calI'(\sigma)$~is not a leaf in the $\D$-model~$\calM'$ and, hence, not reflexive. 

The only remaining new case is the application of the bottom row of Table~\ref{table:ApDT} for a non-$\D$-saturated $\sigma : \Diamond \varphi$, i.e.,~when node~$\sigma$ is a leaf of the sequent tree, in BNUIP(iii)$'$. Let 
\ifarxiv\[\else$\fi{} 
\calM, \calI \not \models \bigspdisj_{i=1}^m \left(\sigma : \Diamond \delta_i \spconj \bigspconj_{\tau \neq \sigma \ast 1} \tau : \gamma_{i,\tau}\right).
\ifarxiv\]\else$\fi{} 
By seriality of~$\calM$, there exists a world $v \in W$ such that~$\calI(\sigma) R v$. 
Then $\calJ \ce \calI' \sqcup \{(\sigma \ast 1, v)\}$ is a multiworld interpretation of $\Gamma'\vlfill{\Diamond \varphi, [\varphi]_{\sigma \ast 1}}$ into~$\calM$ such that 
\ifarxiv\[\else$\fi
\calM, \calJ \not \models \bigspdisj_{i=1}^m \left( \sigma \ast 1 : \delta_i \spconj \bigspconj_{\tau \neq \sigma \ast 1}\tau : \gamma_{i,\tau} \right).
\ifarxiv\]\else$\fi 
By induction hypothesis, there is a multiworld interpretation~$\calJ'$ of $\Gamma'\vlfill{\Diamond \varphi, [\varphi]_{\sigma \ast 1}}$ into some $\D$-model~$\calM'$ such that $(\calM', \calJ')\sim_p(\calM, \calJ)$ and   $\calM', \calJ' \not \models \Gamma'\vlfill{\Diamond \varphi, [\varphi]_{\sigma \ast 1}}$. Similar to the case of  $\Box \varphi$ for~$\K$, restricting this~$\calJ'$ to the labels of~$\Gamma$ yields a multiworld interpretation bisimilar to~$\calI$ and refuting $\Gamma = \Gamma'\vlfill{\Diamond \varphi}_\sigma$.
\end{proof}

\ifarxiv
\section{Uniform interpolation for~{\sf S5}}\label{sec:uipS5}

The uniform interpolation property easily follows for logics satisfying local tabularity and the Craig interpolation property~\cite{ChagZak97BookModalLogic}. A logic is \emph{locally tabular} if there are only finitely many pairwise nonequivalent formulas for each finite set of atomic propositions. Examples of locally tabular logics are classical propositional logic and~$\Sfive$.
In this case, the left interpolant~$\forall p \varphi$ can be taken to be the disjunction of all formulas~$\psi$ without~$p$ implying~$\varphi$ (accordingly, the right interpolant~$\exists p \varphi$ is the conjunction of all formulas~$\psi$ without~$p$ implied by~$\varphi$).

Although proving uniform interpolation for~$\Sfive$ is therefore simple, we want to use our method applied to a hypersequent calculus for~$\Sfive$, which  provides a direct construction for the interpolants. Important for our method are the form of Kripke models and the structure of the proof system. For~$\K$, $\T$,~and~$\D$ we used intransitive treelike models and nested sequents mimicking this treelike structure, which fit well with the recursive step of our method. $\Sfive$~is complete with respect to single finite clusters, i.e.,~finite models with the total accessibility relation. In the rest of this section we only work with these kinds of models, i.e.,~it is assumed that $R =  W \times W$.

Cut-free hypersequent calculi for~$\Sfive$ were first (independently) introduced in~\cite{Avr96,Min71PSIM,Pot83JSL}. A hypersequent has the form $\calG= \Gamma_1 \mid \cdots \mid \Gamma_n$ where $\Gamma_i$'s~are multisets of formulas in negation normal form, 
and its corresponding formula $
\iota(\calG) \ce \Box \big( \bigvee \Gamma_1 \big) \vee \dots \vee \Box \big( \bigvee \Gamma_n \big)$.
We use letters~$\calG$~and~$\calH$ to denote hypersequents.  Among the many existing hypersequent calculi, we use the one closest to tableaus. The
hypersequent rules for~$\Sfive$ used here are presented in Fig.~\ref{fig:rulesS5}. These modal rules  can be found (as derived rules) in~\cite{Fit07MLH}. They are the sequent-style equivalent of what Fitting called there the ``Simple $\Sfive$~Tableau System,'' i.e.,~prefixed tableaus with prefixes being integers rather than sequences of integers, and are used to reduce hypersequent completeness to tableau completeness. The same rules can be obtained by Kleene'ing the $\Sfive$~hypersequent calculus from~\cite{Restall:2007} as explained in~\cite[Sect.~5]{KuzLel15LJIGPL} (strictly speaking, rules in~\cite{KuzLel15LJIGPL} are grafted hypersequent rules for~{\sf K5}, but the crown rules for these grafted hypersequents are exactly the hypersequent rules for~$\Sfive$; another minor difference is that we are using one-sides sequents and negation normal form). Being Kleene'd, these rules form a terminating calculus for~$\Sfive$ under the proviso that ${\sf k}$~and~$\rulet$~be applied only if the principal~$\Diamond \varphi$ in their conclusion is saturated w.r.t.~the component of the active formula~$\varphi$ and that all the other rules are  applied  only when their principal formula is not saturated in the conclusion, as defined presently. 

\begin{figure}[ht]
\centering
  \fbox{\parbox{.7\textwidth}{
    \begin{center}
      $\vlinf{\idTnd}{}{\calG \mid \Gamma, {p,\overline p} }{}$
      \qquad
      $\vlinf{\idTop}{}{\calG \mid \Gamma, {\top} }{}$
      \end{center}  
      \begin{center}
      \quad
      $\vlinf{\lor}{}{\calG \mid \Gamma, \varphi \vee \psi }{\calG \mid \Gamma, \varphi \vee \psi, \varphi, \psi }$
      \qquad
      $\vliinf{\land}{}{\calG \mid \Gamma, \varphi \wedge \psi }{\calG \mid \Gamma, \varphi \wedge \psi, \varphi }{\calG \mid \Gamma, \varphi \wedge \psi, \psi }$
      \end{center}
      \begin{center}
      $\vlinf{\Box}{}{\calG \mid \Gamma, \Box \varphi }{\calG \mid \Gamma, \Box \varphi \mid \varphi}$
      \qquad
      $\vlinf{\mathsf{k}}{}{\calG \mid \Gamma, \Diamond \varphi \mid \Delta }{\calG \mid \Gamma, \Diamond \varphi \mid \Delta, \varphi }$
      \qquad
      $\vlinf{\rulet}{}{\calG \mid \Gamma, \Diamond \varphi }{\calG \mid \Gamma, \Diamond \varphi, \varphi }$
    \end{center}
  }}
\caption{Terminating hypersequent rules for~$\Sfive$}
\label{fig:rulesS5}
\end{figure}

\begin{definition}[Saturation in hypersequents]
A formula~$\theta$ is \emph{saturated} in  a hypersequent~\mbox{$\calH \mid \Gamma, \theta$} if it satisfies the following conditions according to the form of~$\theta$: 
\begin{itemize}
\item
$\theta$~is an atomic formula; 
\item
if $\theta = \varphi \lor \psi$, then both~$\varphi$~and~$\psi$ are in\/~$\Gamma$;
\item
if $\theta = \varphi \land \psi$, then at least one of~$\varphi$~or~$\psi$ is in\/~$\Gamma$;
 \item
if $\theta = \Box \varphi$, then $\varphi$~is either in~$\calH$ or in\/~$\Gamma$;
\end{itemize}
The formula $\theta = \Diamond \varphi$ is saturated with respect to a sequent component of~$\calH$ if $\varphi$~is in that sequent component.
A hypersequent~$\calG$ is \emph{saturated} if all diamond formulas in it are saturated w.r.t.~each sequent component of~$\calG$, all other formulas  are saturated, and, additionally,  $\calG$~is neither of the form $\calH \mid \Gamma, \top $ nor of the form $\calH \mid \Gamma, p,\overline{p} $ for any atomic proposition $p \in {\sf Prop}$.
\end{definition}

\emph{Labels} for hypersequents are natural numbers. For a hypersequent $\calG = \Gamma_1 \mid \cdots \mid \Gamma_n$ we use the set of labels $\calL(\calG) = \{1, \dots, n\}$. We define  multiworld interpretations and multiformulas for hypersequents by analogy with nested sequents, but now using natural numbers as labels.\footnote{Strictly speaking, these labels impose an ordering on the sequent components turning it into a sequence of sequents rather than a multiset of sequents. Since permuting sequent components is both trivial and tedious, we continue with the multiset representation, stating labels explicitly if necessary.}

\begin{definition}\label{def: Multiworld Interpretation}
A \emph{cluster-like multiworld interpretation of a hypersequent~$\calG = \Gamma_1 \mid \cdots \mid \Gamma_n$ into an $\Sfive$-model $\calM = (W,W\times W,V)$}   is a function $\calI : \{1,\dots,n\} \to W$.
\end{definition}

Within this section, by ``multiworld interpretation'' we always mean ``cluster-like multiworld interpretation.'' Note that there is no restriction on the image of~$\calI$, because we work with $\Sfive$-models where all worlds are related to each other. For a fixed multiworld interpretation~$\calI$, we usually write $w_i$ instead of~$\calI(i)$ and represent the whole~$\calI$ by $w_1, \dots, w_n$. A multiworld interpretation $w_1, \dots, w_n$ is injective if the worlds~$w_i$ are pairwise disjoint. The rest of the definitions and results for hypersequents are completely analogous to the nested sequent setting (modulo the change of labels into natural numbers). The analog of Def.~\ref{def:truth_seq} is\looseness=-1

\begin{definition}
Let $\calM$ be a model with worlds $w_1, \dots, w_n$ and let $\calG = \Gamma_1 \mid \dots \mid \Gamma_n$ be a hypersequent. We say that $\calM, w_1, \dots, w_n \models \calG$  if{f} 
\[
\calM, w_i \models \varphi \text{ for some } i  \text{ and } \varphi \in \Gamma_i.
\]
A hypersequent~$\calG$ is \emph{valid in a model~$\calM$}, denoted $\calM \models \calG$, when $\calM, w_1, \dots, w_n \models \calG$ for all multiworld interpretations $w_1, \dots, w_n$ of~$\calG$ into~$\calM$. 
\end{definition}

 We have completeness for the validity of hypersequents, i.e.,~$
\calM \models \calG \text{ if{f} } \calM \models \iota(\calG),
$ 
for all hypersequents~$\calG$ and $\Sfive$-models~$\calM$.

A multiformula is similarly defined as in Def.~\ref{def:multiformula}, where we now use natural numbers as labels instead of sequences of natural numbers, i.e.,~use~$n$ instead of~$\sigma$. All definitions and lemmas about multiformulas based on nested sequents also apply to the hypersequent setting (Def.~\ref{def:interpretationmultiformula} until Lemma~\ref{lem:sdnf}).

Uniform interpolation for hypersequents is defined in the same way as for nested sequents. All definitions and lemmas between Def.~\ref{def:uip nested sequents} and Cor.~\ref{lem:bnuip equiv} are naturally adapted to the hypersequent setting. Instead of NUIP and BNUIP we now speak of the \emph{hypersequent uniform interpolation property~(HUIP)} and the \emph{bisimulation hypersequent uniform interpolation property~(BHUIP)} respectively.

So far, everything goes analogously to the nested sequent case. Even defining the uniform interpolants seems to work analogously. However, when performing the inductive proof (analogous to Theorem~\ref{thm:uipK}) ensuring that those are actual uniform interpolants, one runs into a problem in the recursive case for saturated sequents. Roughly speaking, the problem is caused by the fact that in $\Sfive$-models, the truth of a formula in one world generally depends on all the worlds, including its immediate ``parent.'' Contrast this with treelike models where the truth of a formula in a world is fully determined by its descendants which are disjoint from its parent, as well as from its siblings and their descendants. The reason this feature of cluster-like models is problematic is that changing the valuation of~$p$ in a later recursive call may conflict with valuations of~$p$ necessitated by the preceding one.

To circumvent this problem, we use a special property of~$\Sfive$: every modal formula is $\Sfive$-equivalent to a formula of modal depth~1 (see~\cite[Sect.~5.13]{Fit83}, where Fitting proved this in order to establish Craig interpolation for~$\Sfive$). This means that we can restrict ourselves to formulas where each literal~$q$~or~$\overline{q}$ is under the scope of at most one modality. 
Therefore, after stripping this one modality away,  the resulting formulas are purely propositional, meaning that no further recursive calls are needed and, at the same time, that their truth values depends on the valuation in only one world instead of all worlds in the model. This resolves the aforementioned conflict between recursive calls.\looseness=-1

\begin{table}[t]
  \begin{center}
  \renewcommand*{\arraystretch}{1.4}
    \begin{tabular}{| l l |} \hline
    $\calG$ matches \quad \quad \quad 
    & 
    $A_p(\calG)$ equals
    \\
    \hline
    $\calG' \mid \{\Gamma, \top \}_k $ 
    &
    $k : \top$
    \\
     $\calG' \mid \{\Gamma, p,\overline p \}_k $ 
    &
    $k : \top$
    \\
    $\calG' \mid \Gamma, \varphi \lor \psi $ 
    &
    $A_p(\calG' \mid \Gamma, \varphi, \psi, \varphi \lor \psi )$
    \\
    $\calG' \mid \Gamma, \varphi \land \psi $ 
    &
    $A_p(\calG' \mid \Gamma, \varphi, \varphi \land \psi ) \spconj A_p(\calG' \mid \Gamma, \psi, \varphi \land \psi )$
    \\
    $\calG' \mid \{ \Gamma, \Box \varphi  \}_k $
    &
    
$\bigspconj\nolimits_{i=1}^m \left(k: \Box \delta_i \spdisj \bigspdisj\nolimits_{j \leq k} (j : \gamma_{i,j})    \right)$ where the SCNF of 
    \\
    &  $A_p(\calG' \mid \{ \Gamma, \Box \varphi  \}_k \mid \varphi )$ is $\bigspconj\nolimits_{i=1}^m \left(k+1: \delta_i \spdisj \bigspdisj\nolimits_{j \leq k} (j : \gamma_{i,j}) \right)$
    \\
    $\calG' \mid  \Gamma, \Diamond \varphi  $
    &
    $A_p(\calG' \mid \Gamma, \Diamond \varphi, \varphi )$
    \\
    $\calG' \mid  \Gamma, \Diamond \varphi  \mid \Delta $
    &
    $A_p(\calG' \mid \Gamma, \Diamond \varphi \mid \Delta, \varphi )$\\ \hline
    \end{tabular}
    \end{center}
\caption{Recursive construction of~$A_p(\Gamma)$ for $\Sfive$-hypersequents for~$\calG$ that are not saturated.}
\label{table:ApS5}
\end{table}

So from now on, we only consider hypersequents $\calG = \Gamma_1 \mid \cdots \mid \Gamma_n$, where each $\Gamma_i$~contains only formulas of modal depth~$\leq1$.  With that in mind, we define multiformula interpolants~$A_p(\calG)$ for hypersequents~$\calG$. If $\calG$~is not saturated, $A_p(\calG)$~is defined in Table~\ref{table:ApS5} following the finite proof-search tree of the hypersequent. In particular, $\varphi \lor \psi$, $\varphi \land \psi$,~and~$\Box \varphi$  must be non-saturated; in the rule for~$\Box \varphi$, w.l.o.g.~we assume~$k$ to be the largest label;  the penultimate row  is applied only if $\Diamond \varphi$~is not saturated w.r.t.~its own component; and the last row is only applied if $\Diamond \varphi$~is not saturated w.r.t.~the component containing the displayed~$\Delta$.

For saturated~$\calG$, we define 
\begin{equation}
\label{eq:iiib_S5}
A_p(\calG) \qquad\ce\qquad \bigspdisj_{\substack{ {k : \ell} \in \calG \\\ell \in {\sf Lit} \setminus \{p,\overline{p}\} }} k : \ell \quad\spdisj\quad
1 : \Diamond \forall p\left(\bigvee\nolimits_{ \Diamond \psi \in \calG} \psi\right)
\end{equation}
where $(\forall p) \xi$~represents the uniform interpolant of a propositional formula~$\xi$ w.r.t.~classical propositional logic. Any known algorithm for its computation can be used. The construction of~$A_p(\calG)$ is well-defined because the recursion in Table~\ref{table:ApS5} terminates by the termination of the rules.

\begin{theorem}
Logic\/~$\Sfive$ has the BHUIP. 
\end{theorem}
\label{thm:uipS5}
\begin{proof}
We follow the proof of Theorem~\ref{thm:uipK} showing the three condition for BHUIP. It is easily seen that $A_p(\calG)$~does not contain~$p$ and that its labels are from~$\calG$.

For BHUIP\eqref{NUIP:2}, let $w_1, \dots, w_n$ be a multiworld interpretation of a hypersequent~$\calG$, and of the multiformula~$A_p(\calG)$, into an $\Sfive$-model $\calM = (W,W\!\times\!W,V)$.  We use induction to show \looseness=-1
\[
\calM, w_1, \dots, w_n \models A_p(\calG) \qquad\text{ implies }\qquad \calM, w_1, \dots, w_n \models \calG.
\] 
First we treat some cases from Table~\ref{table:ApS5} and then we consider the case where $\calG$~is saturated.
\begin{itemize}
\item 
Both $\calG \mid \{\Gamma, p, \overline{p}\}_k  $ and $\calG \mid \{\Gamma, \top \}_k $ hold in all models, under all interpretations. 
\item Boolean cases work the same way as for nested sequents.
\item 
The case of~$\Box \varphi$ is also very similar. The only difference from the nested case for~$\K$ is that instead of considering only children of the node where $\Box \varphi$~needs to be true in a treelike model, here we have to consider all worlds in the model. Otherwise, the reasoning is the same. \looseness=-1
\item The penultimate row of Table~\ref{table:ApS5} can be processed the same way as the row for~$\T$ in Table~\ref{table:ApDT} because $\Sfive$-models are similarly reflexive. 
\item The last row of Table~\ref{table:ApS5} works the same way as the last row of Table~\ref{table:Ap} because the interpretation of the label with~$\varphi$ is in both cases accessible from the interpretation of the label with~$\Diamond \varphi$.

\item 
Finally, if $\calG$~is saturated, let 
$
\calM, w_1, \dots, w_n \models A_p(\calG)$ for~$A_p(\calG)$ from~\eqref{eq:iiib_S5}.  As for nested sequents, the case of $\calM, w_1, \dots, w_n \models k : \ell$ with $k : \ell \in \calG$ is straightforward.  It remains to consider the case when, $\calM, w_1\models \Diamond \forall p\left(\bigvee_{\Diamond \psi \in \calG} \psi\right).$ This means that there is a $v \in W$ such that $\calM, v \models \forall p\left(\bigvee_{\Diamond \psi \in \calG} \psi\right)$. Since $\forall p \xi \to \xi$ is a  propositional tautology for any~$\xi$ by Def.~\ref{def:UniformInterpolation}, we have $\calM, v \models \psi$ for some $\Diamond \psi \in \calG$.
Therefore $\calM, w_k \models \Diamond \psi$ for all~$k$, including the label of the component containing~$\Diamond \psi$. Thus, $\calM, w_1, \dots, w_n \models \calG$.
\end{itemize}

For BHUIP(iii)$'$, let $w_1, \dots, w_n$ a multiworld interpretation of~$\calG$ into an $\Sfive$-model $\calM = (W,W\!\times\!W,V)$ such that 
$
\calM, w_1, \dots, w_n \not \models A_p(\calG).
$
We need to find worlds $w'_1, \dots, w'_n$ from another $\Sfive$-model $\calM' = (W',W'\!\times\! W',V')$ such that $(\calM, w_1, \dots, w_n) \sim_p (\calM', w'_1, \dots, w'_n)$ and
$
\calM', w'_1, \dots, w'_n\not \models \calG.
$ We define~$\calM'$ and $w'_1, \dots. w'_n$ and prove BHUIP(iii)$'$ by simultaneous recursion. We first consider the case where $\calG$~is saturated, then we show several cases following Table~\ref{table:ApS5}.
\begin{itemize}
\item For $\calG$~being saturated, we assume $\calM, w_1, \dots, w_n \not \models A_p(\calG)$   for~$A_p(\calG)$ from~\eqref{eq:iiib_S5}.
We have three steps in the construction of model~$\calM'$, which can be compared to the steps of the construction in Theorem~\ref{thm:uipK}.\looseness=-1
\begin{enumerate}[(1)]
\item 
Whenever $w_i=w_j$, duplicate this world, until all~$w_i$'s are distinct. Clearly, this yields a $p$-bisimilar model $\calN=(W',W'\!\times\! W',V_\calN)$ with $W'\supseteq W$ and an injective multiworld interpretation $w'_1,\dots,w'_n$ of~$\calG$ into~$\calN$ such that $\calN, w'_1,\dots,w'_n \not\models A_p(\calG)$. \looseness=-1
\item 
\label{constr:S5_2} 
Now we construct a model~$\calN' $ from~$\calN$ by changing valuations of~$p$ in all worlds $v\notin\{w'_1,\dots,w'_n\}$. It follows from the last disjunct in~\eqref{eq:iiib_S5} that $\calN, v \not\models \forall p\left(\bigvee\nolimits_{ \Diamond \psi \in \calG} \psi\right)$ for all such~$v$. It is a straightforward consequence of Def.~\ref{def:UniformInterpolation} for the purely propositional formula $\bigvee\nolimits_{ \Diamond \psi \in \calG} \psi$ that it is possible to modify the valuation~$V_\calN(p)$ in such a way that for the resulting $\calN'\ce(W',W'\!\times\! W',V'_\calN)$ we have $\calN', v \not\models \bigvee\nolimits_{ \Diamond \psi \in \calG} \psi$ for all worlds $v\notin\{w'_1,\dots,w'_n\}$. 
Changing only truth values of~$p$ results in a $p$-bisimilar model.
\item \label{constr:S5_3} 
Finally, we define model $\calM'\ce(W',W'\!\times\! W',V'_p)$ to be the same as model~$\calN'$ except for valuations of~$p$ as follows:
$V'_p(p) \ce V'_\calN(p)
\sqcup  \{w'_k \mid k : \overline{p} \in \calG\} \setminus \{w'_k \mid k: p \in \calG\}$.
Note that the resulting model is still $p$-bisimilar and, moreover, 
$\calM', v \not \models \bigvee_{\Diamond \psi \in \calG} \psi$ still holds for all $v\notin\{w'_1,\dots,w'_n\}$.
\end{enumerate}
This finishes the construction.

Now we prove that $\calM',w'_k \not \models \varphi$ whenever $k : \varphi \in \calG$ by induction on the structure of~$\varphi$. 
\begin{itemize}
\item 
We leave the cases for~$\top$, $\bot$, $p$, $\overline{p}$, $\psi \lor \psi'$,~and~$\psi \land \psi'$, which are analogous to~$\K$, to the reader. 
\item 
If $k : \Box \psi \in \calG$, then by saturation, there is a label~$l$  such that~$l : \psi \in \calG$. By induction hypothesis, $\calM', w'_l\not \models  \psi$. Therefore, $\calM', w'_k \not \models  \Box \varphi$.
\item 
If $k : \Diamond \psi \in \calG$, then for each $v \in W'$ we have to prove $\calM', v \not \models \psi$. First, consider $v=w'_l$ for some~$l$. Since $\calG$~is saturated, $l : \psi \in \calG$. By induction hypothesis $\calM', w'_l\not \models  \psi$. Otherwise, if $v\notin\{w'_1,\dots,w'_n\}$, the falsity of~$\psi$ was assured in step~\eqref{constr:S5_3}. Thus, $\calM', w'_k \not \models \Diamond \psi$.
\end{itemize}
\end{itemize}

 There is nothing new for non-saturated cases from Table~\ref{table:ApS5}. Most of them work the same way as for~$\K$, with the exception of the penultimate row that works the same way as for~$\T$ and uses reflexivity of~$\Sfive$-models. 
\end{proof}
\fi

\section{Conclusion}
\label{sect:concl}
We have developed a constructive method of proving uniform interpolation based on 
\ifarxiv 
generalized sequent calculi such as nested sequents and hypersequents. 
\else 
nested sequent calculi.
\fi
While this is an important and natural step to further exploit these formalisms, much remains to be done. This method works well for the non-transitive logics~$\K$, $\D$,~and~$\T$ but meets with difficulties, e.g.,~for~$\Sfive$, which is also known to enjoy uniform interpolation. 
\ifarxiv
And while we successfully adapted the method to hypersequents to cover this logic, the adaptation relies on the reduction to uniform interpolation for classical propositional logic and, thus, is not fully recursive. 
\fi
There are other logics in the so-called modal cube between $\K$ and $\Sfive$ with the UIP, for which it remains to find the right formalism and adaptation of our method. Another natural direction of future work is intermediate logics, where exactly seven logics are known to have the UIP. \looseness=-1

\bibliographystyle{plainurl}
\providecommand{\noopsort}[1]{}

\end{document}